\documentclass[a4paper,10pt]{article}

\usepackage{a4wide}
\usepackage{titling}
\usepackage{xcolor}
\usepackage{graphicx}
\usepackage{comment}
\usepackage{datetime}
\usepackage[noend]{algpseudocode}
\usepackage{algorithm}
\usepackage{amsmath}
\usepackage{siunitx}
\usepackage{balance}
\usepackage{hyperref}
\usepackage{cleveref}
\usepackage{authblk}
\usepackage{amssymb,amsmath,amsthm}

\newtheorem{lemma}{Lemma}

\author[1,3]{Alex Shamis}
\author[2]{Matthew Renzelmann}
\author[1]{Stanko Novakovic}
\author[4]{Georgios Chatzopoulos\footnote{Work done while at Microsoft Research}}
\author[5]{Anders T. Gjerdrum\raise0.35ex\hbox{{\footnotesize *}}}
\author[6]{Dan Alistarh\raise0.35ex\hbox{{\footnotesize *}}}
\author[1]{Aleksandar Dragojevi\'{c}}
\author[1]{Dushyanth Narayanan}
\author[1]{Miguel Castro}

\affil[1]{Microsoft Research}
\affil[2]{Microsoft}
\affil[3]{Imperial College London}
\affil[4]{EPFL}
\affil[5]{UIT: The Arctic University of Norway}
\affil[6]{IST Austria}

\newenvironment{smitemize}%
  {\begin{list}{$\bullet$}%
     {\setlength{\parsep}{0pt}%
      \setlength{\topsep}{0pt}%
      \setlength{\itemsep}{0pt}}}%
  {\end{list}}

  {\begin{list}{\theenumi .}%
     {\usecounter{enumi}%
      \setlength{\parsep}{0pt}%
      \setlength{\topsep}{0pt}%
      \setlength{\itemsep}{0pt}}}%
  {\end{list}}

\newenvironment{captiontext}{%
   \begin{center}%
     \begin{minipage}{0.9\linewidth}%
         \footnotesize}%
   {%
      \end{minipage}%
        \end{center}}

\newcommand{\unit}[1]{\,{#1}}

\newlength{\figurewidth}
\newlength{\widefigurewidth}
\newlength{\figbottrim}
\newlength{\graphbottrim}
\newcommand{\figtrimvals}{\setlength{\figbottrim}{8cm}\setlength{\graphbottrim}{2cm}}
\newcounter{para}[subsection]

\date{}

\begin{document}

\sloppypar

\def\baseline{BASELINE}
\def\opaque{FaRMv2}
\def\farmorig{FaRMv1}
\def\farmopaque{FaRMv2}

\title{\textbf{Fast General Distributed Transactions with Opacity using Global Time}} 

\maketitle

\begin{abstract}
Transactions can simplify distributed applications by hiding data
distribution, concurrency, and failures from the application
developer. Ideally the developer would see the abstraction of a single
large machine that runs transactions sequentially and never
fails. This requires the transactional subsystem to provide {\em
  opacity} (strict serializability for both committed and aborted
transactions), as well as transparent fault tolerance with high
availability. As even the best abstractions are unlikely to be used if
they perform poorly, the system must also provide high performance.

Existing distributed transactional designs either weaken this
abstraction or are not designed for the best performance within a data
center. This paper extends the design of FaRM --- which provides
strict serializability only for committed transactions --- to provide
opacity while maintaining FaRM's high throughput, low latency, and
high availability within a modern data center. It uses timestamp
ordering based on real time with clocks synchronized to within tens of
microseconds across a cluster, and a failover protocol to ensure
correctness across clock master failures. FaRM with opacity can commit
5.4 million neworder transactions per second when running the TPC-C
transaction mix on 90 machines with 3-way replication.
\end{abstract}

\section{Introduction}

Cloud data centers provide many relatively small, individually
unreliable servers. Cloud services need to run on clusters of such
servers to maintain availability despite individual server
failures. They also need to scale out to increase throughput beyond
that of a single server. For latency-sensitive applications that need
to keep data in main memory, scale-out is also required to go beyond
the memory limits of a single server.

The challenge is that distributed applications, especially stateful
ones, are much harder to program than single-threaded or even
multi-threaded applications. Our goal is to make them easier to
program by providing the abstraction of a single large machine that
runs transactions one at a time and never fails.  This requires a
distributed transactional system with the following
properties:
\begin{smitemize}
\item {\em Serializability}: All executions are equivalent to
some serial ordering of committed transactions.
\item {\em Strictness}: This ordering is consistent with real
  time.
\item {\em Snapshot reads}: All transactions see a consistent snapshot of the
  database until they commit or abort.
\item {\em High availability}: The system recovers transparently
from server failures and downtimes are short enough to appear as
transient dips in performance.
\end{smitemize}
The combination of the first three properties is also referred to as {\em
  opacity}~\cite{opacity}. Intuitively, opacity extends the properties
of strict serializability to aborted transactions, i.e., these
transactions also see a consistent snapshot at a point in time
consistent with real-time ordering, until they abort.

As even the best abstractions are unlikely to be used if they perform
poorly, the system must also provide scalability and high
performance. Existing designs either weaken this abstraction or are
not designed for the best performance within a data
center. Spanner~\cite{spanner} is a geo-distributed database that
provides opacity with availability but does
not provide low latency and high throughput in the data center.
Several transactional
systems~\cite{farm_sosp,drtm_eurosys,namdb,fasst} have leveraged large
amounts of cheap DRAM per server, fast commodity networking hardware,
and RDMA to achieve good performance in the data center. RDMA can
improve networking throughput and latency by orders of magnitude
compared to TCP~\cite{farm_nsdi}. One-sided RDMA can reduce CPU costs
and latency further compared to two-way messaging using RDMA, as it
bypasses the remote CPU. Several distributed transactional
protocols~\cite{farm_sosp,drtm_eurosys,namdb} use one-sided RDMA to
send fewer messages and achieve higher performance than two-phase
commit (2PC) within the data center.

Current designs that use one-sided RDMA do not provide opacity.
\farmorig{}~\cite{farm_sosp,farm_nsdi} and
DrTM~\cite{drtm_sosp,drtm_eurosys} provide scalability,
availability, and strict serializability for committed but not for
aborted transactions. Optimistically executing transactions in these
systems might read inconsistent state with the guarantee that such
transactions would eventually abort. NAM-DB~\cite{namdb} provides read
snapshots but not strictness, serializability, or high availability.

In this paper, we describe \farmopaque{}, which extends the original
design and implementation of \farmorig{} to provide read
snapshots to all transactions. FaRMv2 uses a novel timestamp-ordering
protocol that leverages the low latency of RDMA to synchronize clocks.
Timestamps are based on real time, which scales well as it allows
machines to use their local clocks to generate timestamps. However,
since clocks are not perfectly synchronized, the transaction protocol
must ``wait out the uncertainty'' when generating read and write
timestamps, which introduces latency.  \farmopaque{} leverages
low-latency, CPU-efficient RDMA-based communication to synchronize
clocks frequently over the network to achieve uncertainties in the
tens of microseconds, two orders of magnitude lower than in
Spanner~\cite{spanner}. Unlike Spanner, \farmopaque{} does not require
atomic clocks or GPS. Instead, servers use the CPU cycle counter and
synchronize with a clock master elected from all the servers in the
system. Timestamp ordering is maintained across clock master failures
using a clock master failover protocol. Our design and implementation
also supports multi-versioning which improves the performance of
read-only transactions. Old versions are kept in memory with efficient
allocation and garbage collection.

The paper makes the following novel contributions:
\begin{smitemize}
\item A mechanism to synchronize clocks to within tens of microseconds
by leveraging RDMA.
\item A transaction protocol with opacity that uses global time and one-sided RDMA.
\item An informal proof of correctness for the transaction protocol.
\item A clock failover protocol that keeps timestamps monotonic
  across clock master failures without requiring special hardware such
  as atomic clocks or GPS.
\item An efficient thread-local, block-based allocator and
  garbage collector for multi-versioning.
\item Other uses of global time for object allocation and for reducing memory overhead at on-disk backups. 
\end{smitemize}

\farmopaque{} can commit 5.4 million neworder transactions per second
when running the TPC-C transaction mix on a cluster of 90 machines
with 3-way replication for fault tolerance. It retains the high
availability of \farmorig{} and can recover to full throughput within
tens of milliseconds of a server failure.  We believe \farmopaque{}
has the highest known throughput for this transaction mix of any
system providing opacity and high availability.

\farmopaque{}'s performance and simple programming model provide a good
platform for developing interactive, large-scale, distributed applications.
We describe how the A1 graph database, which is part of Microsoft's Bing
search engine, is built on \farmopaque{} in Section~\ref{sec:a1}.

\section{Motivation}
\label{sec:motiv}

Serializability is an easy isolation level for developers to
understand because it avoids many anomalies. We found that {\em
  strictness} and {\em opacity} are also important for developers
using a transactional platform.

Strict serializability~\cite{strict-serializability} means that the
serialization order of transactions corresponds to real time. If $A$
completes before $B$ starts, then any correct execution must be
equivalent to a serial one where $A$ appears before $B$. This is
important when clients communicate using some channel outside the
system as is often the case, for example, when other systems are
layered on top of the database.

Opacity~\cite{opacity} is the property that transaction executions are
strictly serializable for aborted transactions as well as committed
transactions.  This simplifies programming by ensuring that invariants
hold during transaction execution. Many existing systems provide
opacity either by using pessimistic concurrency control with read
locks (e.g., Spanner~\cite{spanner}), or by using timestamp ordering
to provide read snapshots during execution (e.g.,
Hekaton~\cite{hekaton_vldb12,hekaton_sigmod13}).  But many systems
that use optimistic concurrency control (OCC)~\cite{opt-cc} do not
provide read snapshots for aborted transactions, e.g.,
\farmorig{}~\cite{farm_sosp,farm_nsdi} and
DrTM~\cite{drtm_sosp,drtm_eurosys}.  This design decision can
improve performance but it imposes a large burden on developers.
Since developers cannot assume invariants hold, they must program
defensively by checking for invariants explicitly in transaction code.
Relational databases reduce this burden by providing mechanisms to
check constraints automatically after each SQL statement, but this can
add a non-trivial performance overhead and it still requires the
developer to write all the relevant constraints.

FaRM and DrTM provide a low level transactional memory model rather
than a relational model. This allows developers to write highly
performant C++ code to manipulate arbitrary pointer-linked data
structures in transactions, but this flexibility comes at a price. It
is not feasible to do efficient automatic constraint checking after
every C++ statement.  Additionally, lack of opacity can lead to
violations of memory safety.  For example, a transaction could read
memory that has been freed and reused, which could cause a crash or an
infinite loop. We illustrate the difficulty of programming without
opacity by discussing the implementation of hash table and B-tree
indices on top of FaRM.

The FaRM hash table~\cite{farm_nsdi} uses chained associative
hopscotch hashing. Each key lookup reads two adjacent array buckets and
zero or more overflow buckets. \farmorig{} ensures that each of these objects 
is read atomically, but they may not all be read from
the same consistent snapshot because \farmorig{} does not ensure opacity.
This can lead to several anomalies, for example,  
a concurrent transaction could move
key $A$ from an overflow bucket to an array bucket while deleting key
$B$, causing the lookup transaction to incorrectly miss $A$. \farmorig{}
solves this problem by adding 64-bit incarnations to all object
headers, replicating them in all overflow bucket pointers, and adding
additional version fields to each array bucket. This
adds complexity and overhead, which could be avoided by
providing opacity.

The FaRM B-Tree implementation keeps cached copies
of internal nodes at each server to improve performance.  Fence
keys~\cite{fencekeys_tods81,fencekeys_vldb04} are used to check the
consistency of parent-to-child traversals.  Strict serializability is
maintained by always reading leaf nodes uncached and adding them to
the read set of the transaction.  The cached internal nodes are shared
read-only across all threads without making additional thread-local or
transaction-local copies. This is extremely efficient as most lookups
require a single uncached read, and do not make copies of
internal nodes. However,
lack of opacity can lead to several anomalies when using this B-tree,
for example, one developer reported that they had found a bug because
``They inserted a key into a B-Tree, looked up the same key in the
same transaction, and it was not there.'' On investigation, we found
that this was possible when a concurrent transaction created a split
in the B-Tree that migrated the key in question ($A$) to a new leaf
object, deleted $A$, and the server running the original transaction
evicted some internal nodes on the path from the root to the new leaf
from the cache. Even though the original transaction would have
aborted in this case, the programmer still needs to reason about
execution before a transaction aborts.  Reasoning about complex corner
cases like this one is hard.  Opacity simplifies programming by
providing strong isolation guarantees even for transactions that
abort.

The main contribution in this paper is adding opacity to FaRM to improve
programmability while retaining good performance.  It is hard to
quantify the benefit of providing a better developer experience. Based
on our deployment experience --- more than two years of FaRMv1
followed by more than two years of FaRMv2 --- we can say that our
developers praised the addition of opacity and we no longer see bug
reports due to opacity violations.  We were also able to remove the
additional version fields per array bucket in the hash table and
convert the ``fat pointers'' for overflow chains to normal pointers,
which simplified the code and reduced space usage.

\section{Background}

\subsection{FaRM}

FaRM~\cite{farm_sosp,farm_nsdi} provides a transactional API to access
objects in a global flat address space that pools together the memory
of a cluster. The API is exposed as library calls, and both
application code and FaRM run within the same process on each machine.
Within a transaction, the application can allocate, free, read and
write objects regardless of their location in the cluster, as well as
execute application logic.  The thread executing the code for a
transaction also acts as the coordinator for the distributed commit of
that transaction. The execution model is symmetric: all threads in a
cluster can be coordinators and all servers can hold in-memory
objects.

FaRM objects are replicated using primary-backup replication. The unit
of replication is a {\em region} (e.g., 2\unit{GB}). All objects in
a region have the same primary and backup servers.  

FaRM implements optimistic concurrency control to enable using
one-sided RDMA to read objects from remote primaries during
transaction execution.  Locking remote objects would require using the
remote CPU or using additional atomic RDMA operations.  So no locks
are taken during transaction execution.  Writes are buffered within
the transaction context. At the end of the execution phase, the
application calls {\sc commit} invoking the commit protocol.  The
commit protocol integrates concurrency control and replication for
fault-tolerance to achieve lower message counts and fewer round trips
than approaches which build distributed commit on top of
replication~\cite{spanner}. The commit protocol first locks write sets
at their primary replicas and then {\em validates} read sets to ensure
serializability.

FaRM has transparent fault tolerance with high availability through
fast failure detection, a reconfiguration protocol for adding/removing
machines, parallel transaction recovery after failure, and
background data re-replication to restore replication levels. Unlike
traditional 2PC, FaRM does not block transactions when a coordinator
fails: coordinator state is recovered in parallel from logs on
participants.

\subsection{One-sided RDMA}

CPU is the
bottleneck when accessing in-memory data using the fast networking
hardware deployed in data centers today.
So FaRM uses one-sided RDMA operations, which are
handled entirely by the remote NIC, to improve performance.  
Remote objects are read using RDMA reads from their primary replica
during transaction execution and read set validation uses RDMA reads of object
versions from the primary.  Unlike traditional 2PC protocols,
primaries of read-only participants do no CPU work in FaRM because RDMA
requests are served by their NICs. Backups of read-only participants
do no work on the CPU or on the NIC. Backups of write-set objects do
not do any CPU work on the critical path of the
transaction; coordinators do a single one-sided RDMA write to each
backup to commit a transaction, and only wait
for the hardware acknowledgement from the NIC. This commit message is
processed asynchronously by the backup's CPU.

There has been a lively debate on the 
merits of one-sided RDMA~\cite{herd,fasst,sequencer,erpc,hybrid-rdma}.
The key issue is that one-sided operations 
and deployed congestion control mechanisms~\cite{dcqcn,timely}
require the reliable connected (RC) mode of RDMA. 
This requires per-connection (queue pair) state which grows with
the size of the cluster, and degrades performance when the state
cannot be cached on the NIC and must be fetched from host
memory instead. 
Sharing connection state across cores can improve
scalability but adds CPU synchronization costs~\cite{farm_nsdi}.

An alternative approach, eRPC~\cite{erpc}, uses connectionless
unreliable datagrams (UD). These scale better at the NIC level because a
single endpoint (queue pair) can send and receive to all servers in
the cluster~\cite{herd,fasst,erpc}.  This approach requires two-sided
messaging, as one-sided RDMAs are not supported over UD. It uses an
RTT-based congestion control mechanism~\cite{timely} implemented in
software. 

The scalability of RC has been improving with
newer NICs.  The RDMA performance of the Mellanox CX3 starts dropping at 256
connections per machine. More recent NICs (CX4, CX5, and CX6) have
better scalability. We measured the scalability of RDMA reads on CX4
RoCE NICs between two machines connected by a 100\unit{Gbps} switch.  We
emulated larger clusters by increasing the number of queue pairs per
machine. We compared this with the performance of 64-byte reads over
eRPC using one queue pair per thread, which is 15 million reads/s. The
RDMA throughput is 35 million reads/s with 28 queue pairs per machine,
and RDMA reads perform better than eRPC with up to 3200 queue pairs
per machine, where they equal eRPC performance. We do not have CX5
hardware and could not measure the performance of 64-byte reads on
CX5. With a small number of queue pairs, RDMA writes~\cite{erpc} and
reads~\cite{anuj_personal} on CX5 have up to 2x higher throughput than
eRPC for reads between 512 bytes and 32\unit{KB}. For larger transfers
both approaches are limited by the line rate.

These results and other recent work~\cite{hybrid-rdma} show that
one-sided RDMA can provide a significant performance advantage for
moderate size clusters.  So we added opacity to FaRM while
retaining all the one-sided RDMA optimizations.

\section{Design}

\subsection{Global time}

Using one-sided RDMA reads during execution makes it challenging to
provide opacity with scalability.
Pessimistic concurrency control schemes such as Spanner~\cite{spanner}
provide opacity by using read locks but this
requires two-way messaging and remote CPU usage on read-only
participants.

Timestamp ordering enables opacity by allowing transactions to read a
consistent snapshot defined by a read timestamp.
The challenge is to generate timestamps scalably and with global
monotonicity, i.e. the timestamp order must match the real time order
in which the timestamps were generated across all servers.

Centralized sequencers do not scale to our target transaction rates.
A state of the art centralized sequencer without fault
tolerance~\cite{sequencer} can generate 122 million timestamps per
second. \farmorig{} can execute 140 million TATP transactions per
second on 90 machines~\cite{farm_sosp}.

NAM-DB~\cite{namdb} uses a centralized timestamp server and caches
timestamps at servers to avoid generating new timestamps per
transaction. This improves scalability but it means that timestamps
are not globally monotonic: timestamps generated on different servers
will not be in real time order. Using non-monotonic timestamps as
transaction read and write timestamps violates strictness.

Clock-SI~\cite{clocksi} has no centralized time server but uses
loosely synchronized physical clocks on each server. Remote reads in
Clock-SI block until the remote server's clock moves past the
transaction read timestamp. This operation is not supported on RDMA
NICs and requires two-sided messaging. Clock-SI is also not globally
monotonic.

Spanner~\cite{spanner} uses Marzullo's algorithm~\cite{marzullo_time}
to maintain globally synchronized real time. Servers synchronize their
local clocks with a time master periodically. The algorithm accounts
for synchronization uncertainty explicitly by representing time as an
interval, which is computed using the round trip time of
synchronization requests, the master time returned by requests, and an
assumed bound on the rate drift of the local clock. Time masters use
atomic clocks and/or GPS that are synchronized with global real time,
and there are multiple clock masters for fault tolerance.  Servers
synchronize with clock masters every 30 seconds. Spanner's
uncertainties are 1--7\unit{ms}, which is acceptable in a
geo-distributed database where latencies are dominated by WAN round
trips.

These uncertainties are too high for FaRM, which is designed for
sub-millisecond transaction latencies using low latency RDMA networks
to scale out within a data center. We also did not want to depend on
atomic clocks and GPS as they are not widely available in all data
centers.

We also use timestamp ordering based on real time, with clock
synchronization using Marzullo's algorithm~\cite{marzullo_time} but
with a design and implementation that provide an average uncertainty
below 20 microseconds, two orders of magnitude lower than in Spanner.
We use only the cycle counters present on all CPUs, allowing any
server in a cluster to function as a clock master (CM) without
additional hardware.

Non-clock masters periodically synchronize their clocks with the
current clock master using low-latency, CPU-efficient RPCs based on
RDMA writes~\cite{farm_nsdi}. Round-trip times are in the tens of
microseconds even under load and a single clock master can handle
hundreds of thousands of synchronization requests per second while
also running the application without noticeable impact on
application throughput. 

Marzullo's algorithm assumes an asynchronous network. When a non CM
fetches the time from a CM over such a network, it can only assume
that the one-way latencies of the request and response messages are
non-negative. The true time at the CM (as computed at the non-CM) thus
lies in an interval defined by the two extreme possibilities shown in
Figure~\ref{fig:sync_drift}. The time elapsed between sending the
request and receiving the response is measured on the local clock of
the non CM. Marzullo's algorithm assumes that clock rates can differ
across machines but that the maximum relative difference, or {\em
  drift}, is bounded by some known value $\epsilon$. The upper bound
on the time must therefore be adjusted as shown in the figure.

\begin{figure}
    \centering
    \figtrimvals
    \includegraphics[trim=0 5cm 6cm 0,clip,width=\textwidth]{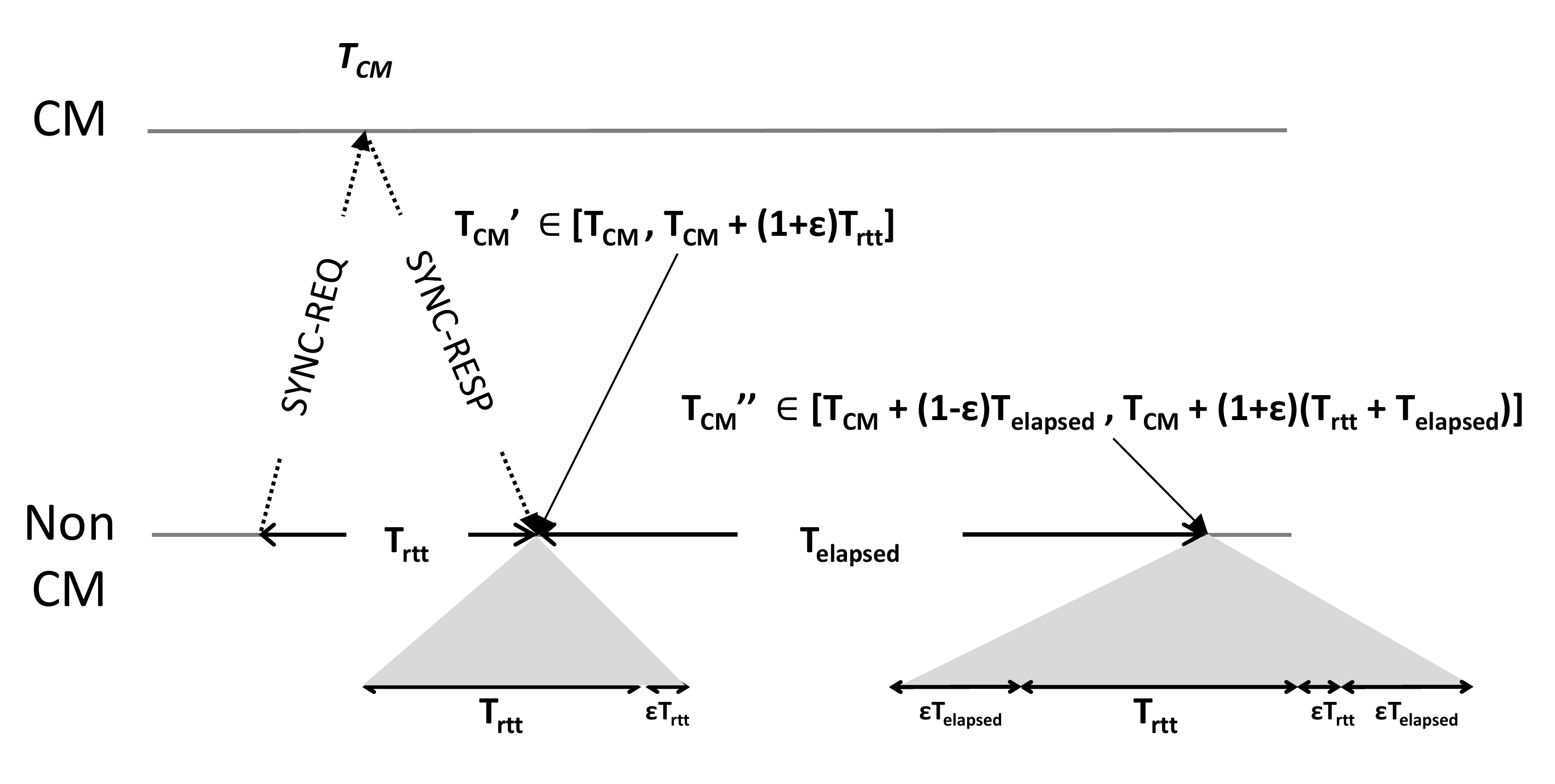}
    \caption{Synchronization and uncertainty in \farmopaque{}}
    \label{fig:sync_drift}
\end{figure}

Once time has been successfully synchronized, this information can be
used to compute a time interval using only the local clock, without
network synchronization. However, both upper and lower bounds must be
extended to account for potential drift; this additional uncertainty
grows linearly with the time elapsed since
synchronization. Figure~\ref{fig:sync_drift} shows how a non CM
computes bounds on the time at the CM ($T_{CM}'$)
immediately after a synchronziation as well as after some time
has elapsed since a synchronization ($T_{CM}''$).

For simplicity, the figure only shows an uncertainty interval computed
from the most recent synchronization. However, this is not necessary
--- any past synchronization in the current configuration can be used
to compute the current time interval. Moreover, the best
synchronization to use (for the tightest bounds) is not always the
most recent, and is not always the same for the upper bound and lower
bound. We used an optimized variant of the algorithm that computes the
optimal (tightest upper and lower bounds) time interval based on all
available information from past synchronizations.  The algorithm is
implemented efficiently by maintaining the state from up to two past
synchronizations: one that gives the highest lower bound ($S_{lower}$) and one that
gives the lowest upper bound ($S_{upper}$). Figure~\ref{fig:marzullo} shows this
optimized algorithm.

\begin{figure}
\begin{algorithmic}
\Function{LB}{$S, T$}           \Comment Compute lower bound
    \State \Return $S.T_{CM} + (T - S.T_{recv})(1-\epsilon)$
\EndFunction
\Function{UB}{$S, T$}           \Comment Compute upper bound
    \State \Return $S.T_{CM} + (T - S.T_{send})(1+\epsilon)$
\EndFunction
\Function{SYNC}{}               \Comment Synchronize with clock master
   \State $S_{new}.T_{send} \gets$ \Call{LOCALTIME}{$ $}
   \State $S_{new}.T_{CM} \gets$ \Call{MASTERTIME}{$ $} \Comment RPC to clock master
   \State $T_{now} \gets$ \Call{LOCALTIME}{$ $}
   \State $S_{new}.T_{recv} \gets T_{now}$

   \If{\Call{LB}{$S_{new}, T_{now}$} > \Call{LB}{$S_{lower}, T_{now}$}}
       \State $S_{lower} \gets S_{new}$  \Comment Update to improve lower bound
   \EndIf

   \If{\Call{UB}{$S_{new}, T_{now}$} < \Call{UB}{$S_{upper}, T_{now}$}}
       \State $S_{upper} \gets S_{new}$\Comment Update to improve upper bound
   \EndIf
\EndFunction
\Function{TIME}{}               \Comment Compute time interval based on synchronization state
   \State $T_{now} \gets$ \Call{LOCALTIME}{$ $}
   \State $L \gets$\Call{LB}{$S_{lower}, T_{now}$}
   \State $U \gets$\Call{UB}{$S_{upper}, T_{now}$}
   \State \Return $\left[L,U\right]$
\EndFunction
\end{algorithmic}
    \caption{\farmopaque{}'s synchronization algorithm}
    \label{fig:marzullo}
\end{figure}

Synchronization state is kept in a shared memory data structure on
each machine. Any thread can read this state to compute a time
interval using the function {\sc TIME}. We also allow any thread to synchronize with the clock
master and update the state. For threads that also run application
code, these requests are interleaved non-preemptively with execution
of application code. This leads to scheduling delays in processing
synchronization requests and responses causing high uncertainty. Hence
in practice, we only send synchronization requests from a single,
high-priority thread that does not run application code. \farmorig{}
already uses such a high-priority thread to manage leases for failure
detection, and we use the same thread to send synchronization requests
in \farmopaque{}.

We assume that cycle counters are synchronized across threads on a
server to some known precision. On our platform, the OS
(Windows) synchronizes the cycle counters
 across threads to within 1024 ticks (about 400\unit{ns}). This
uncertainty is included in the interval returned to the
caller.

Time intervals in \farmopaque{} are globally monotonic: if an event at
time $\left[L_1,U_1\right]$ anywhere in the cluster happens-before an
event at time $\left[L_2,U_2\right]$ anywhere in the cluster, then
$U_2 > L_1$. We also guarantee that on any thread, the left bound $L$
is non-decreasing.

Frequent synchronization allows us to use conservative clock drift
bounds.  Currently we use a bound of 1000 parts per million (ppm),
i.e. 0.1\%. This is at least 10x higher than the maximum allowed by
the hardware specification on our servers and at least 10x higher than
the maximum observed rate drift across 6.5 million server hours on our
production clusters with more than 700 machines.

Correctness in \farmopaque{} requires clock frequencies to stay within
these bounds. We use the local CPU clock on each machine, which on
modern hardware is based on extremely accurate and reliable crystal
oscillators. In rare cases, these can be faulty at manufacture: we
detect these cases using an initial probation period when a server is
added to the cluster, during which we monitor clock rates but do not
use the server. Clock rates can also change slowly over time due to
aging effects and temperature variation (e.g.,~\cite{clock_aging}). 
\farmopaque{} continuously monitors
the clock rate of each non-CM relative to the CM. If this exceeds
200\unit{ppm} (5x more conservative than the bound we require for correctness), 
it is reported to a centralized service that removes
either the non-CM or the CM, if the CM is reported by multiple servers.

\subsection{\farmopaque{} commit protocol}

Figure~\ref{fig:commit_protocol} shows \farmopaque{}'s 
transaction protocol as a time diagram for one example
transaction. The line marked $C$ shows
the coordinator thread for this transaction.  The other lines show
other servers with primaries and backups of objects accessed by the
transaction. \farmopaque{} uses primary-backup replication. In this
example the transaction reads two objects $O_1$ and $O_2$ and
writes a new value $O_1'$ to $O_1$. Each object is replicated on one
primary and one backup.  The backup for $O_2$ is not shown as backups
of objects that are read but not written do not participate in the
protocol.

\begin{figure}
    \centering
    \figtrimvals
    \includegraphics[trim=0 \figbottrim{} 0 0,clip,width=\textwidth]{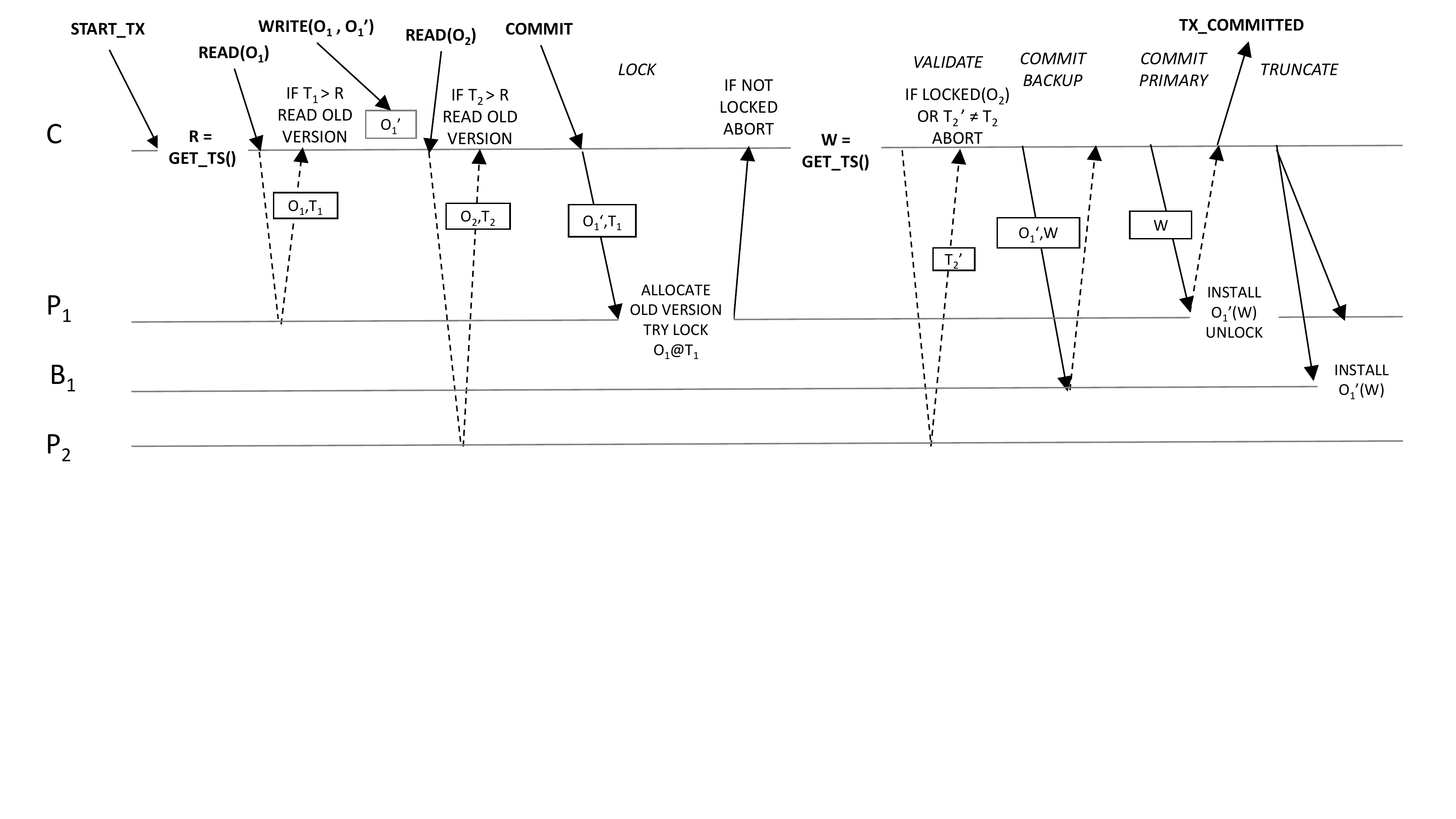}
    \begin{captiontext}
      Execution and commit timeline for a transaction that reads two objects and
      writes one of them. Solid arrows show RDMA writes.
      Dashed arrows show RDMA reads and hardware acks for RDMA
      writes.
    \end{captiontext}
    \caption{\farmopaque{} commit protocol}
    \label{fig:commit_protocol}
\end{figure}

Transactions obtain a read timestamp when they start executing and
transactions that modify objects obtain a write timestamp when they
commit. \farmopaque{} serializes transactions in timestamp order:
read-only transactions are serialized by their read timestamp and
read-write transactions are serialized by their write timestamp.

The application starts a transaction by calling {\sc start\_tx}, which
acquires the read timestamp $R$. The reads of the transaction are then
performed at time $R$, i.e., successful reads see the version with
the highest write timestamp that is less than or equal to $R$. The time
at the CM must exceed $R$ before the first read can be issued, to
ensure that no future writes can have a timestamp less than or equal
to $R$. This is necessary for opacity. The time at the CM must also be
less than or equal to $R$ at the time {\sc start\_tx} is invoked. This is
necessary for strictness, to avoid reading a stale snapshot.
\farmopaque{} satisfies both conditions by setting $R$ to the upper
bound of the time interval when {\sc start\_tx} is called and waiting out
the uncertainty in this interval before the first read is issued
(Figure~\ref{fig:acquire_ts}).

\begin{figure}
\begin{algorithmic}
\Function{GET\_TS}{}
    \State $\left[L,U\right] \gets$ \Call{TIME}{$ $} \Comment $U$ is in the future here.
    \State \Call{SLEEP}{$(U-L)(1+\epsilon)$} \Comment Wait out uncertainty.
    \State \Return $U$ \Comment $U$ is now in the past.
\EndFunction
\end{algorithmic}
    \begin{captiontext}
    $\epsilon$ is the clock drift bound.
    \end{captiontext}
    \caption{Timestamp generation (strict serializability)}
    \label{fig:acquire_ts}
\end{figure}

By including an ``uncertainty wait'' when acquiring a read 
timestamp $R$, we ensure that the time at the clock master equalled
$R$ at some point during the call to $GET\_TS$. Intuitively
an uncertainty wait blocks the transaction until the time interval
at the end of the wait no longer overlaps the time interval
at the beginning of the wait (Figure~\ref{fig:uncertainty_wait}).

\begin{figure}
    \centering
    \figtrimvals
    \includegraphics[trim=0 12cm 15cm 0,clip,width=\textwidth]{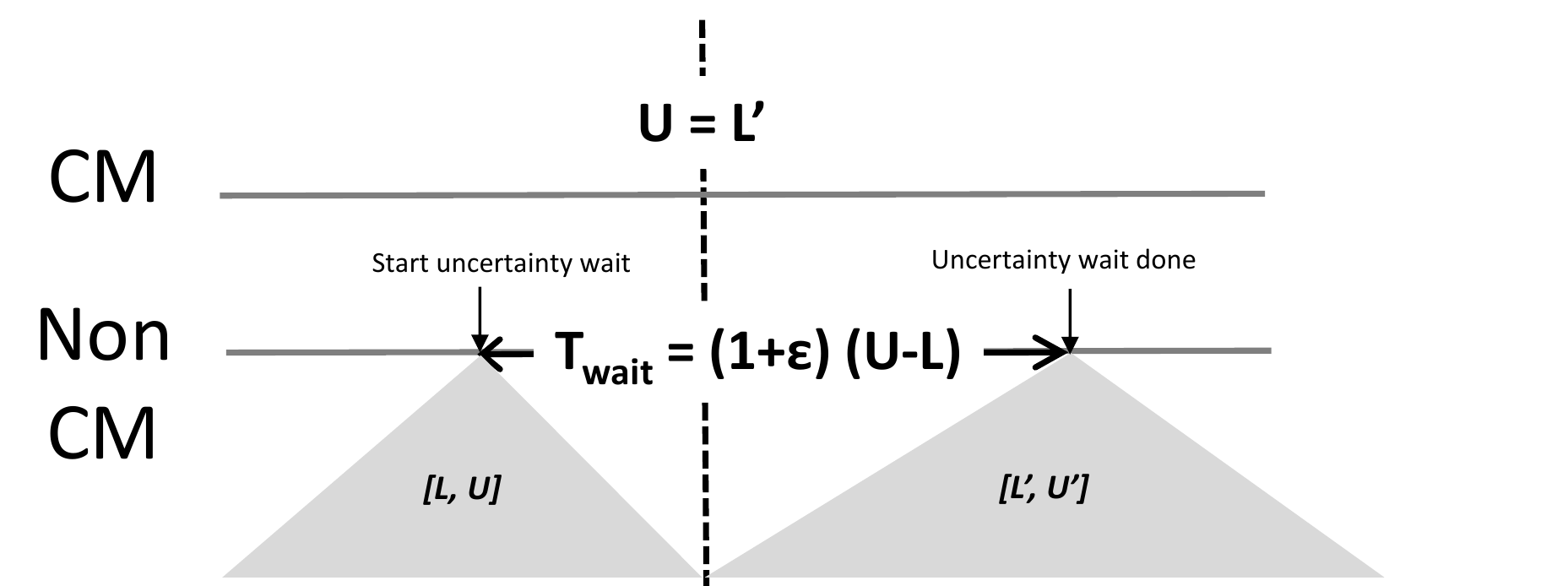}
    \caption{Uncertainty wait in \farmopaque{}}
    \label{fig:uncertainty_wait}
\end{figure}

The application can then issue reads and writes for objects in the
global address space.  Reads are always from the primary replica. If
the read object has a timestamp greater than $R$, \farmopaque{} will
find and return the correct snapshot version of the object (not shown
in the figure).

Writes to objects are buffered locally.  When the application calls
{\sc commit}, the coordinator sends {\sc lock} messages to all
primaries that have objects in the write set. If the versions of all
these objects equal the versions read and they could be successfully
locked, the coordinator acquires a write timestamp $W$. $W$ must be
greater than the time at the CM when all write locks are taken, to
ensure that any conflicting reader with read timestamp $R' \geq W$
will see either the lock or the eventual commit. $W$ must also be less
than the time at the CM when read validation begins, to ensure that
the transaction will see any conflicting writer with write timestamp
$W' \leq W$. Both conditions are satisfied by computing $W$ using {\sc GET\_TS} to wait out the uncertainty
(as for the read timestamp): acquiring a timestamp in
the future and then waiting until the timestamp has gone into the past.

Note that we do not send {\sc lock} messages to backups. This makes
writes more efficient but it means that reads of the latest version of
an object must always be from the primary, preventing reads from
backups for load balancing. In FaRM, we rely on sharding for
load-balancing: each physical machine is the primary for some shards
and the backup for other shards and thus reads to different shards can
be balanced across different machines.

The coordinator then validates objects that were read but not
written, using RDMA reads to re-read object versions from the
primaries.  Validation succeeds only if all such objects are unlocked
and at the version read by the transaction.  The coordinator then
sends {\sc commit-backup} messages to the backups of the write-set objects
using RDMA writes. When all {\sc commit-backup} messages have been
acknowledged by the remote NICs, the coordinator sends
{\sc commit-primary} messages to the primaries of the write-set objects
using RDMA writes.  Finally, {\sc truncate} messages are sent to write-set
primaries and backups to clean up per-transaction state maintained on
those servers. These are almost always piggybacked on other messages.
On receiving {\sc commit-primary}, primaries install the new versions of
write-set objects and unlock them. Backups install the new values
when they receive {\sc truncate}.

This protocol retains all the one-sided RDMA optimizations in
\farmorig{}'s commit protocol with no added communication.  As in
\farmorig{}, the protocol can also exploit locality, i.e. co-locating
the coordinator with the primaries. If all primaries are co-located
with the coordinator then only {\sc commit-backup} and piggybacked
{\sc truncate} messages are sent and there is no remote CPU
involvement at all on the critical path.

In addition, \farmopaque{} skips validation for all read-only
transactions, which was not possible in \farmorig{}. Committing a read-only
transaction in \farmopaque{} is a no-op. This can reduce the number of
RDMA reads issued by read-only transactions by up to 2x compared to
\farmorig{}.

By default \farmopaque{} transactions are strictly serializable.
\farmopaque{} also supports snapshot isolation (SI) and non-strict
transactions.  It does not support weaker isolation levels than
snapshot isolation. Non-strictness and SI can be set per transaction:
developers need only use this for transactions where it will improve
performance significantly without affecting application correctness.
SI and non-strictness are implemented with small changes to the
protocol shown in Figure~\ref{fig:commit_protocol}.

SI transactions in \farmopaque{} skip the validation phase. In
\farmorig{}, validation was required to check whether objects that
were read but not written, were read from a consistent snapshot.  In
\farmopaque{}, consistent snapshots are already provided during
execution. Validation is require for serializability (to check that
the read snapshot is still current at the write timestamp) but not for
snapshot isolation. Additionally, SI transactions do not need to
perform the write uncertainty wait with locks held. Instead the wait
is done concurrently with the {\sc commit-backup} and {\sc
  commit-primary} messages. This reduces latency and also contention
(by reducing the time for which locks are held).

Strictness can be relaxed in \farmopaque{} both for serializable and
SI transactions. For SI transactions, strictness means that if
transaction $A$ starts after transaction $B$, then the read timestamp
of $A$ is greater than or equal to the write timestamp of $B$.
Non-strict transactions choose the lower bound $L$ on the time
interval $\left[L,U\right]$ as the read timestamp when {\sc start\_tx}
is called, without any uncertainty wait.  Non-strict SI transactions
compute their write timestamp as their upper bound $U$ of the time
interval at the point of write timestamp acquisition, again without
any uncertainty wait. The uncertainty wait for the write timestamp is
required for serializable read-write transactions, whether strict or
non-strict.

\subsection{Fast-forward for failover}

\farmopaque{} maintains global monotonicity across clock master
failures.  When the clock master fails or is removed, a new clock
master is chosen.  As we do not rely on externally synchronized time
hardware such as atomic clocks or GPS, the new clock master must
continue based on a time interval obtained by synchronizing with a
previous clock master. Adding the uncertainty in this interval to all
time intervals generated in the future would maintain monotonicity,
but uncertainty could grow without bound.

\farmopaque{} shrinks the uncertainty on a new clock master during
clock master failover. It uses a fast-forward protocol that we
integrated with FaRM's reconfiguration protocol~\cite{farm_sosp} which
is used to add or remove servers from the cluster.

\begin{figure}
    \centering
    \figtrimvals
    \includegraphics[trim=0 7cm 0 0,clip,width=\textwidth]{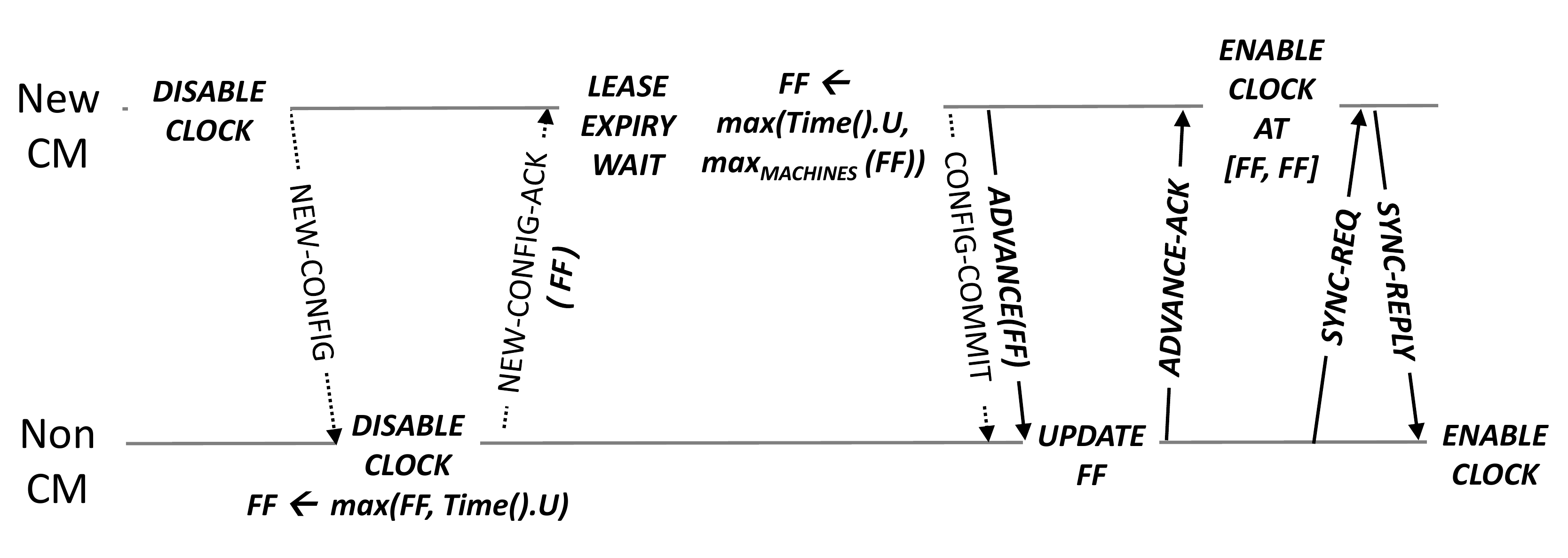}
    \begin{captiontext}
    $Time()$ returns the current time interval $\left[L, U\right]$ on the
    coordinator.
    \end{captiontext}
    \caption{Clock recovery after clock master failure}
    \label{fig:reconfiguration}
\end{figure}

In FaRM, a configuration is identified by a unique sequence number and
specifies the membership of the cluster, with one member distinguished
as the ``configuration manager'' (CM) for the configuration. Leases
are used for failure detection. Each non-CM maintains a lease at the
CM and the CM maintains a lease at each non-CM. Lease renewals are
periodically initiated by each non-CM via a 3-way handshake which
renews both the non-CM's and the CM's leases.

Configurations are stored in Zookeeper and are changed by an atomic
compare and swap that 
increments the sequence number and installs the new
configuration atomically. A configuration change is initiated by the current CM
if it suspects a non-CM has failed or a new server joins the cluster,
and by a non-CM if it suspects the CM has failed. The server
initiating the configuration change, if successful, becomes the new
CM. After committing the configuration change to Zookeeper, the new CM
uses a 2-phase protocol to commit the new configuration to all the
servers in the new configuration. This mechanism handles single and
multiple node failures as well as network partitions. If there is a
set of connected nodes with the majority of the nodes from the
previous configuration, with at least one replica of each object,
and with at least one node that can update the configuration in Zookeeper, 
then a node from this partition will become the CM. Otherwise,
the system will block until this condition becomes true.

In our design, the CM also functions as the clock master.  This lets
us reuse messages in the reconfiguration protocol for clock master
failover, and lease messages for clock synchronization.  Figure~\ref{fig:reconfiguration} shows the
reconfiguration/fast-forward protocol. Dotted lines show existing
messages in the reconfiguration protocol; bold lines show messages
that were added for clock failover. For simplicity, the figure omits
the interaction with Zookeeper (which remains unchanged) and shows
only one non-CM.  To implement fast-forward, each server maintains a
local variable $FF$ marking the last time clocks were fast-forwarded.
The new CM first disables its own clock: the clock continues to
advance, but timestamps are not given out and synchronization requests
from other servers are rejected. It then sends a {\sc new-config}
message to all non-CMs.  On receiving {\sc new-config}, each non-CM
disables its clock and sets $FF$ to the maximum of its current value
and the upper bound on the current time interval. This updated value
is piggybacked on the {\sc new-config-ack} sent back to the CM.

After receiving all acks from all non-CMs, the CM waits for one lease
expiry period which ensures that servers removed from the
configuration have stopped giving out timestamps. The CM then advances
$FF$ to the maximum of the upper bound of its current time
interval and the maximum $FF$ on all servers in the new configuration
including itself. It then commits the new configuration by sending
{\sc config-commit}, sends $FF$ to all non-CMs ({\sc advance}), and
waits for acknowledgements ({\sc advance-ack}). After receiving acks
from all non-CMs, the CM enables its clock with the time interval set
to $\left[FF, FF\right]$.  The additional round trip to propagate $FF$
ensures that time moves forward even if the new CM fails immediately
after enabling its own clock.  After receiving {\sc config-commit},
non-CMs send periodic synchronization requests to the new CM. On the
first successful synchronization, a non-CM clears all previous
synchronization state, updates the synchronization state and enables
its clock.

This protocol disables clocks for up to three round trips plus one
lease expiry period (e.g., 10\unit{ms}) in the worst case.  While clocks are disabled,
application threads requesting timestamps are blocked. If the old CM
was not removed from the configuration, then it remains the CM and we
do not disable clocks or fast-forward. If only the old CM was removed
(no other servers fail), the ``lease expiry wait'' is skipped as we
know that the old CM's lease has already expired. Adding new servers
to the configuration does not disable clocks. The worst-case clock
disable time is incurred when the CM and at least one non-CM fail at
the same time.

Fast-forward can cause the \farmopaque{} clock to diverge from
external time. If synchronization with external time is important,
this can be done by ``smearing'', i.e., gradually adjusting clock rates
without at any point violating the drift bound assumptions. Currently
\farmopaque{} does not do smearing as there is no need for
timestamps to match external time.

\subsection{Multi-versioning}

\begin{figure}
    \centering
    \figtrimvals
    \includegraphics[trim=0 \figbottrim{} 4cm 0,clip,width=\textwidth]{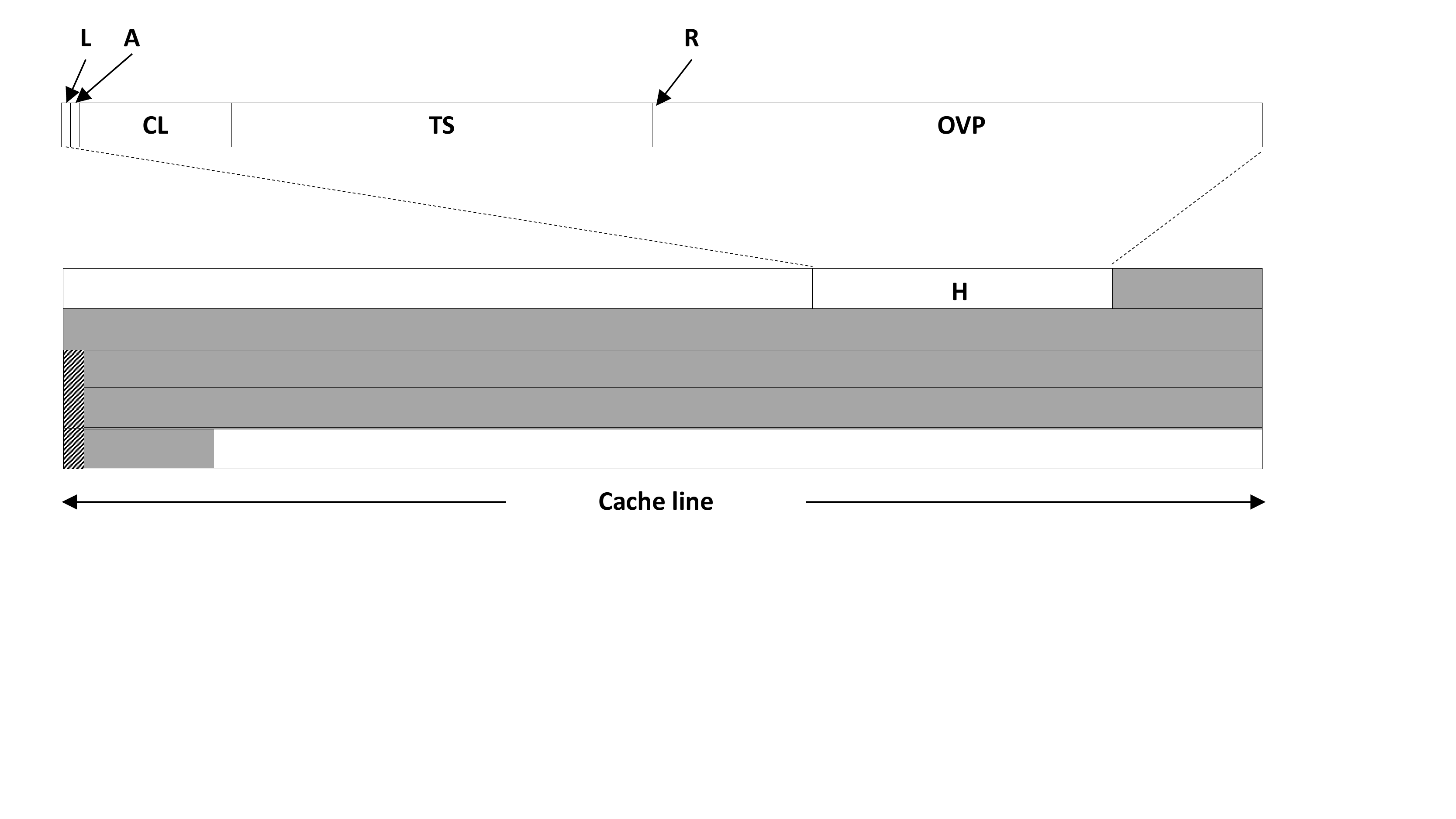}
    \begin{captiontext}
    The figure shows a \farmopaque{} object laid out across multiple cache
    lines with header ($H$), application data (in gray), and cache line versions
    (hashed). The top part of the figure shows an expanded view of the header.
    \end{captiontext}
    \caption{\farmopaque{} object layout}
    \label{fig:object_layout}
\end{figure}

\farmopaque{} supports multi-version concurrency
control~\cite{Bernstein:1986:CCR:17299}. This can help reduce
aborts in read-only transactions caused by conflicting writes.
Multi-versioning is implemented using a per-object, in-memory linked
list of old versions in decreasing timestamp order. This is optimized
for the case where most objects have no old versions, and those that
do have only a few. These are reasonable assumptions for our current
production workloads which have interactive read-only queries and
transactional updates but not batch analytics that take hours or days.
For example, we have built a distributed graph database,
A1~\cite{a1_hpts} on \farmopaque{} which must serve graph queries
within 50\unit{ms} as part of Microsoft's Bing search engine.

Reading an object always starts at the {\em head version} whose
location does not change during the lifetime of the object. This
ensures that we can always read the head version using a single
one-sided RDMA without indirection.

Each head version in \farmopaque{} has a 128-bit header $H$ which
contains a lock bit $L$, a 53-bit timestamp {\em TS}, and (in
multi-version mode) an old-version pointer {\em OVP}
(Figure~\ref{fig:object_layout}). {\em TS} contains the write
timestamp of the last transaction to successfully update the
object. $L$ is used during the commit protocol to lock objects. {\em
  OVP} points to a singly linked list of older versions of the object,
ordered by decreasing timestamp. The fields $L$, $A$, $CL$, and $R$
are used as in \farmorig{}~\cite{farm_nsdi,farm_sosp}.  For example,
$A$ is a bit that indicates whether the object is currently allocated, and $CL$ is an 8-bit counter that is incremented when a new version of the object is installed and is used to ensure RDMA reads observe writes performed by local cores atomically. The value of $CL$ is repeated at the start of each cache line after the first (as shown in Figure~\ref{fig:object_layout}).

If the head version timestamp is beyond the transaction's read
timestamp, the linked list is traversed, using RDMA reads if the
version is remote, until a version with a timestamp less than or equal
to the read timestamp is found.  Old versions also have a 128-bit
object header but only the {\em TS} and {\em OVP} fields are used. Old
versions are allocated from globally-addressable, unreplicated,
RDMA-readable regions with primaries co-located with the primaries of
their head versions.

Old version memory is allocated in 1\unit{MB} blocks carved out of
2\unit{GB} regions. Blocks are owned and accessed by a single thread
which avoids synchronization overheads. Each thread has a currently
active block to which allocation request are sent until the block runs
out of space.  Within the block we use a bump allocator that allocates
the bytes in the block sequentially.  Allocating an old version thus
requires one comparison and one addition, both thread-local, in the
common case.

An old version is allocated when a thread on a primary receives a
{\sc lock} message. The thread allocates space for the old version, locks
the head version, copies the contents as well as the timestamp and
old-version pointer of the head version to the old version, and then
acknowledges to the coordinator.  When {\sc commit-primary} is received,
a pointer to the old version is installed at the head version before
unlocking. As allocation is fast, the dominant cost of creating old
versions is the memory copy. This copy is required in order to keep
the head version's location fixed, which lets us support one-sided RDMA
reads.

\subsection{Garbage collection}

\farmopaque{} garbage-collects entire blocks of old-version memory
without touching the object headers or data in the block. A block is
freed by the thread that allocated it to a thread-local free block
pool. No synchronization is required unless the free block pool becomes
too large (at which point blocks are freed to a server-wide pool).

Each old version $O$ has a GC time that is equal to the write
timestamp of the transaction that allocated $O$. If the transaction
that allocated $O$ aborted, then the GC time of $O$ is zero. The GC
time of a block is the maximum GC time of all old versions in the
block. It is kept in the block header and updated when transactions commit.
A block can be freed and reused when its GC time is less than
the GC safe point {\em GC}, which must be chosen such that no transaction
will attempt to read old versions in the block after it is freed.

\begin{figure}
    \centering
    \figtrimvals
    \includegraphics[trim=0 \figbottrim{} 9cm 0,clip,width=\textwidth]{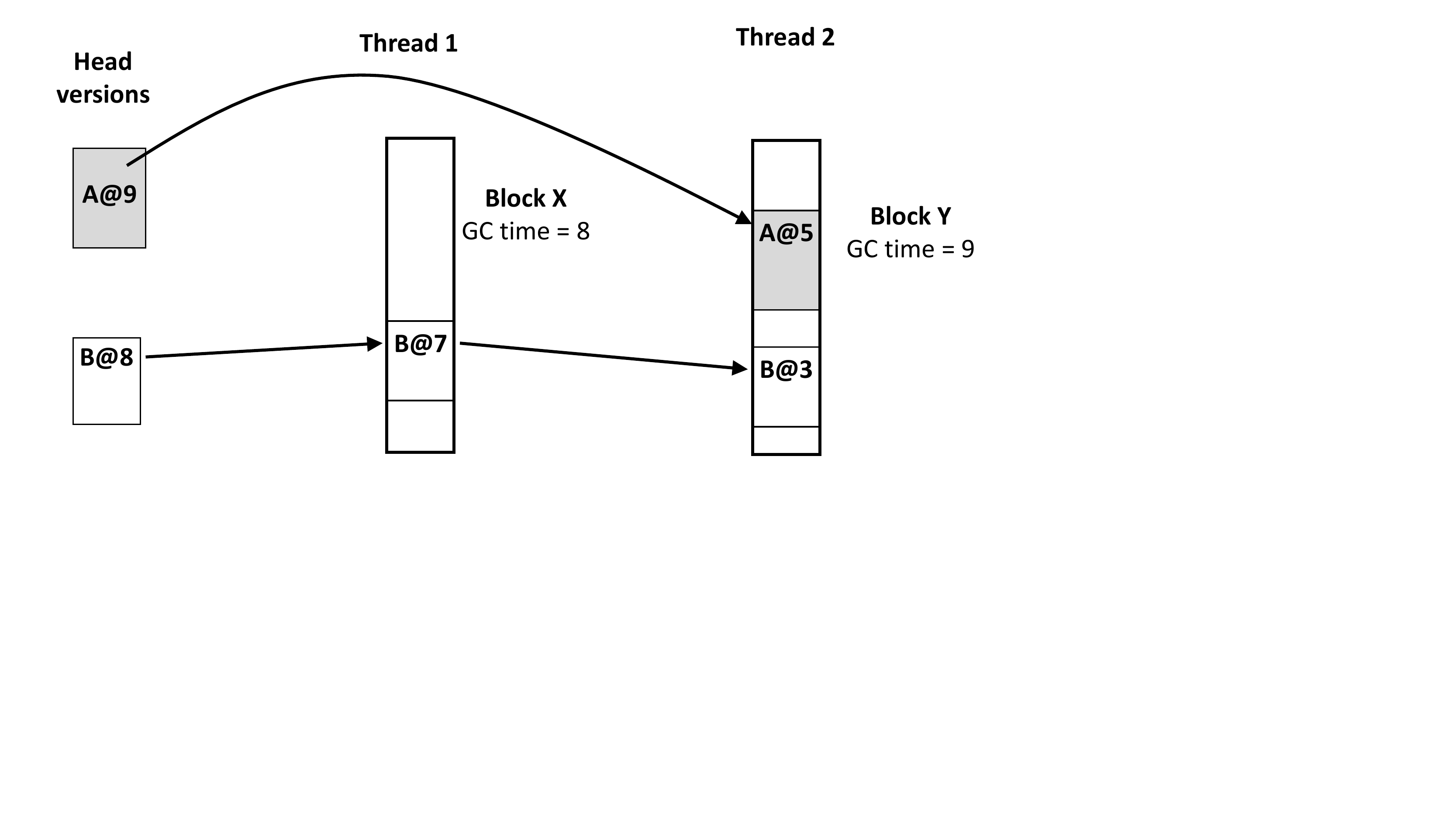}
    \caption{Old versions in \farmopaque{}}
    \label{fig:gc}
\end{figure}

Figure~\ref{fig:gc} shows a simple example with two objects $A$ and
$B$ with one and two old versions respectively, in two blocks owned by
two different threads. The list for a given object is in decreasing
timestamp order and the list pointers can cross block and thread
boundaries arbitrarily.  Old versions of the same object can be freed
out of order and the list is not guaranteed to be null-terminated. 

We compute {\em GC} to ensure that readers never follow a pointer to an old
version that has already been freed. \farmopaque{} uses the periodic lease renewal
messages to propagate information about the read timestamp of the oldest active transaction
({\em OAT}) in the system. The value of {\em OAT} is used to compute  {\em GC} .
Each thread maintains a local list of currently
executing transactions with that thread as coordinator, and tracks
the oldest read timestamp of these: this is the per-thread {\em OAT}. When
sending a lease request, each non-CM includes the minimum of this
value across all threads, and of the lower bound of the current time
interval: this is the per-machine {\em OAT}.  The CM tracks the
per-machine {\em OAT} of all machines in the configuration including
itself. The global {\em OAT} is the minimum of all of these values and is
included in all lease responses to non-CMs.

Each non-CM thus has a slightly stale view of the global {\em OAT}, which
is guaranteed to be less than or equal to the global {\em OAT} and will
catch up to the current value of the global {\em OAT} on the next lease
renewal. The global {\em OAT} is guaranteed to be less than or equal to
the read timestamp of any currently running transaction in the system.
It is also guaranteed to be less than or equal to the lower bound $L$
of the time interval on any thread in the system. $L$ is
non-decreasing on every thread and transactions created in the future,
both strict and non-strict, will have a read timestamp greater than
$L$.

\subsection{Parallel distributed read-only transactions}
\label{sec:dist_ro}

\farmopaque{} supports parallel distributed read-only transactions, which can
speed up large queries by parallelizing them across the cluster and
also partitioning them to exploit locality. To support these,
\farmopaque{} supports stale snapshot reads, which are read-only
transactions that can execute with a specified read timestamp $R$
which can be in the past, i.e., less than the current time lower bound
$L$.  A master transaction that acquires read timestamp $R$ can send
messages to other servers to start slave transactions at time $R$,
ensuring that all the transactions execute against the same
snapshot. As $R$ may be in the past when the message is received, this
requires stale snapshot reads.

Without stale snapshot reads, we could use {\em OAT} directly as the
GC safe point {\em GC}.  However this is not safe in the presence of
stale snapshot reads.  The coordinator of a master transaction with
read timestamp $R$ can fail after sending a message to execute a slave
transaction but before the slave transaction has been created. The
global {\em OAT} could then advance beyond $R$, causing a block to be
freed that is then read by a slave transaction, violating
opacity. \farmopaque{} uses a second round of propagation to solve
this problem. 

\begin{figure}
    \centering
    \figtrimvals
    \includegraphics[trim=0 4cm 5cm 0,clip,width=\textwidth]{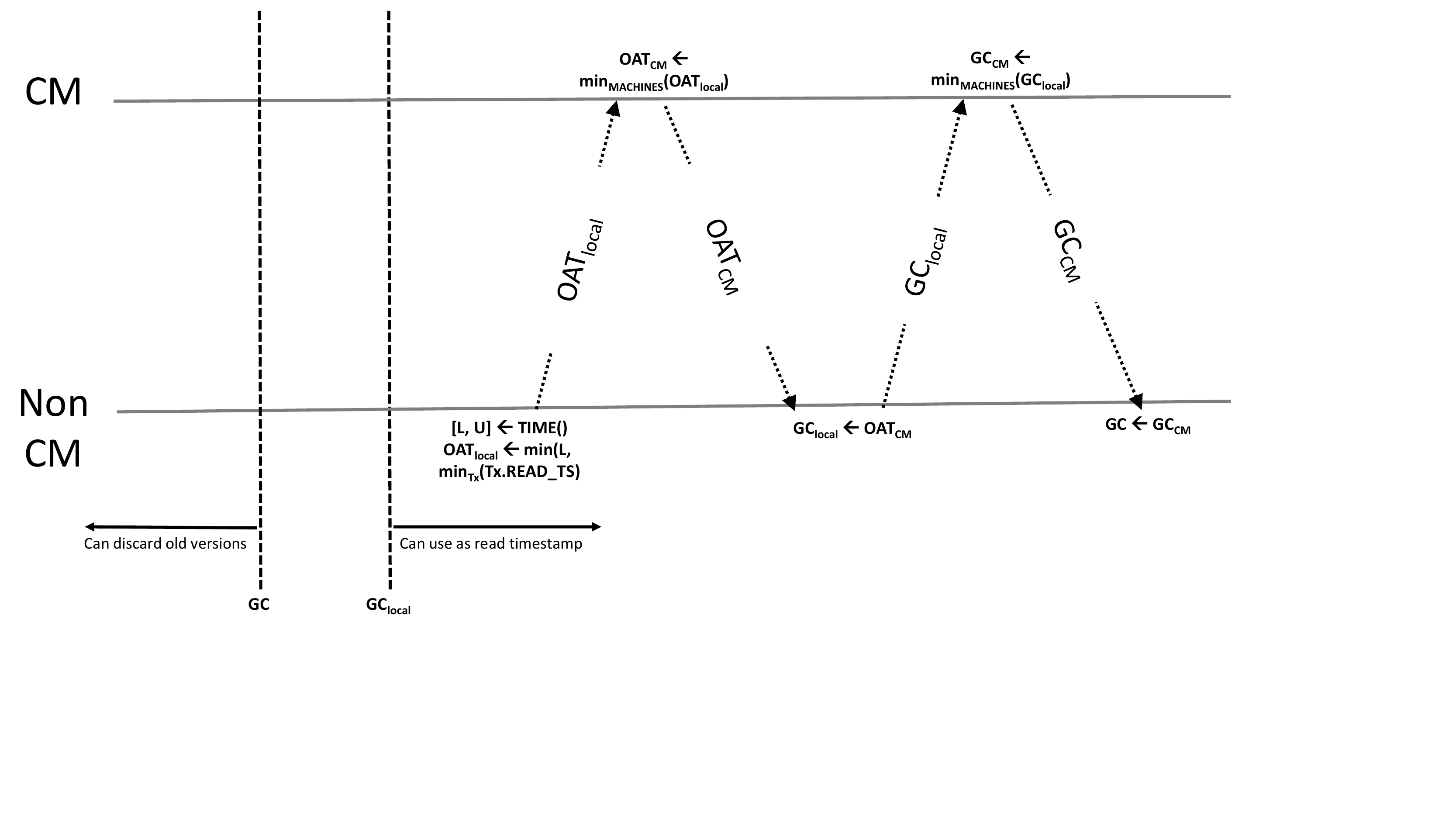}
    \caption{GC point computation in \farmopaque{}}
    \label{fig:gc_point}
\end{figure}

Figure~\ref{fig:gc_point} shows how we compute the safe GC point {\em
  GC}.  As described before, machines periodically propagate $OAT_{local}$,
which is the minimum of the current lower bound on the time and the
read timestamp of any active local transaction. This is sent to the CM
piggybacked on lease request messages. The CM computes $OAT_{CM}$ as
the mininum $OAT_{local}$ value across all servers, and piggybacks
this on lease responses. This value is used at each server to update a
$GC_{local}$ value. A second round of propagation (piggybacked on the next pair of lease request and response messages) provides each server with a 
$GC$ value.

We disallow slave transactions with read timestamps $R < GC_{local}$.
This guarantees safety as $GC_{local}$ on any server is guaranteed to
be greater than or equal to $GC$ on any server, and therefore the
slave transaction will never attempt to read freed memory.  This also
guarantees that slave transaction creation always succeeds as long as
the master transaction co-ordinator has not failed.

\subsection{Early aborts}

\farmopaque{} keeps primary but not backup copies of old versions.
With $n$-way replication, this reduces the space and CPU overhead of
multi-versioning by a factor of $n$. When a primary for a region fails and a
backup is promoted to primary, the new primary will have no old
versions for objects in that region. Readers attempting to read the
old version of an object $O$ will abort. This is a transient condition: if the
transaction is retried it will be able to read the latest version of
$O$, which exists at the new primary.  As we expect failures to be
relatively infrequent we allow such ``early aborts'' during failures
in return for reduced common-case overheads.

Unlike the opacity violations in \farmorig{} described at the
beginning of this section, these early aborts do not require
significant additional developer effort. Developers need not write
code to check invariants after every read, detect infinite loops, or
handle use-after-free scenarios: all of which were required when using
\farmorig{}. The code must already handle aborts during the commit phase due
to optimistic concurrency control, e.g., by retrying the
transaction. Early aborts can be handled in the same way.

Allowing early aborts also lets us do {\em eager
  validation}~\cite{conf/wdag/RiegelFF06} as a performance optimization.  If a
serializable transaction with a non-empty write set attempts to read
an old version, then \farmopaque{} fails the read and aborts the
transaction even if the old version is readable, as the transaction
will eventually fail read validation. We also allow applications to
hint that a serializable RW transaction is likely to generate writes;
in this case we abort the transaction when it tries to read an old
version even if the write set is currently empty.

For some workloads, old versions are accessed so infrequently that the
cost of multi-versioning outweighs the
benefit. \farmopaque{} can operate in single-version mode. In this
mode, it does not maintain old versions even at the primaries.

\subsection{Slab reuse}
\label{sec:slab-reuse}

Another use of global time is in the memory allocator for \farmopaque{} objects.
\farmopaque{} inherits \farmorig{}'s efficient slab allocator
mechanism~\cite{farm_sosp} for allocating head versions of objects.
These are allocated by the transaction co-ordinator during transaction
execution when the application issues a call to allocate an object.  Objects
are allocated from 1\unit{MB} slabs, with all objects within a slab
having the same size.  This enables a compact bitmap representation of
the free objects in a slab, and we use a hierarchical bitmap structure
to find a free object in the bitmap with only a small number of memory accesses. Each
slab is owned by a single thread on the machine holding the primary
replica of the slab. Thus the common case for allocation requires no
messages and no thread synchronization but simply accesses
thread-local state on the transaction co-ordinator.

The allocation state of an object is also maintained as a single bit
in the object header. This bit is updated at both primary and backup
replicas during transaction commit. The free object bitmap is only
maintained at the primary. On failure, when a backup is promoted to
primary of a region, it scans the object headers in live slabs in the region to build the free object
bitmap for those slabs.

Over time, as the distribution of allocated object sizes changes, we
may find a large number of empty slabs with object size $S_{old}$ and
insufficient slabs with some other object size $S_{new}$. Thus it is
important to be able to reuse empty slabs with a different object size.
As object timestamps and old version pointers are inlined into each
object header for efficiency, we must ensure that transactions do not
access the old headers after the slab has been reused; the header
locations may then contain arbitrary object data, violating opacity and
memory safety. Global time and the {\em OAT} mechanism in
\farmopaque{} enable a simple and efficient solution to this that we
describe below.

\begin{figure}
    \centering
    \figtrimvals
    \includegraphics[trim=0 12cm 0 0,clip,width=\textwidth]{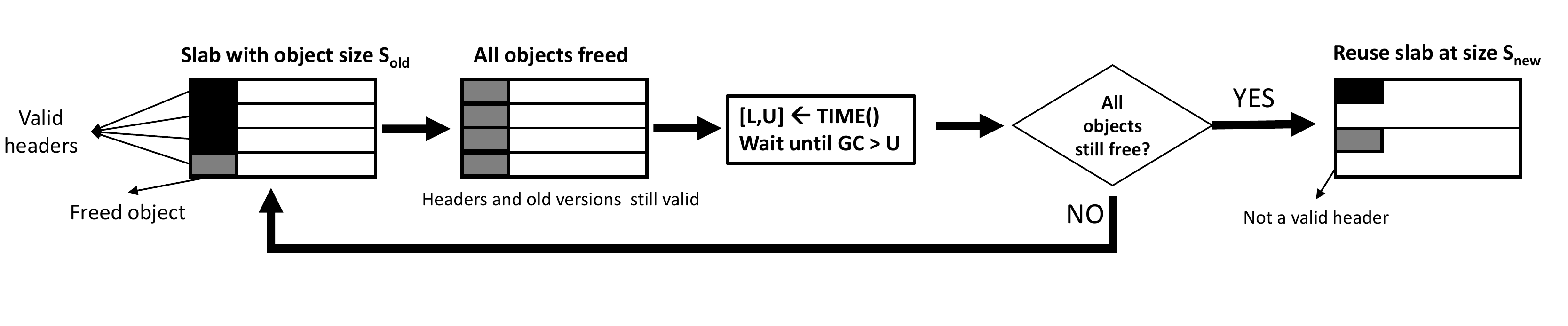}
    \caption{Slab reuse in \farmopaque{}}
    \label{fig:slab_reuse}
\end{figure}

\farmopaque{} uses the GC safe point to ensure that a slab is freed
only when no transactions can be accessing any version of an object in
the slab. When all objects in a slab are marked free, the primary
replica of the slab records the current time $\left[L, U\right]$.  It
then waits for {\em GC} to pass $U$.  If an allocation from the slab
occurs during this wait, the slab is not freed but continues to be
used with the current object size.  Otherwise, after the wait, we are
guaranteed to have seen all allocations and frees on the slab, to have
no transactions reading from the slab, and to have no allocated
objects in the slab. At this point the primary informs both backups
that the slab is free, and then adds the slab to a free list whence it
can be re-allocated with a different object
size. Figure~\ref{fig:slab_reuse} shows slab reuse in \farmopaque{}.

Threads at backup replicas apply transaction updates asynchronously and potentially out of order; this happens when the in-memory transaction logs on
the backups are truncated. We ensure that all updates to a slab are
truncated and applied at all replicas before freeing the slab. To do
this, the co-ordinator of each transaction considers the transaction
active (and therefore includes it in the computation of $OAT_{local}$)
until the transaction is truncated at all replicas and acknowledged to
the co-ordinator. Truncation requests and acknowledgements are
piggybacked on existing messages most of the time. Under low load when there are
no messages to piggyback on, we ensure that $OAT_{local}$ (and
hence {\em GC}) still advances by sending explicit truncation messages
triggered by a timer.

\subsection{Backups on disk}
\label{sec:backups-on-disk}

Another use of global time is to reduce memory overhead when backups are kept on disk.
\farmopaque{}, like \farmorig{}, can place backup replicas on disk (or
SSD) rather than in DRAM. This reduces the cost per gigabyte of the
storage (by using less DRAM) at the cost of slower update performance
and slower data recovery on failure. The read performance is
unaffected as the primary replica is still maintained in DRAM.

Data on disk backups is maintained in a log-structured format, with
committed transactions writing all updated objects to the log.  A
background process cleans the log by copying live data between log
extents. When a primary replica for a \farmopaque{} region fails, a
new primary must be rebuilt by scanning the log for that region on the
backups.  This is done in parallel across threads, machines, and
regions being recovered, but can still take hundreds of milliseconds
to seconds.

For high availability during this bulk recovery, FaRM supports {\em
  on-demand disk reads} from backups, for objects that are not yet
available at the primary. This requires an in-memory {\em redirection
  map} from each object address to the location of its most recently
committed version in the disk log.

\begin{figure}
    \centering
    \figtrimvals
    \includegraphics[trim=0 4cm 0 0,clip,width=\textwidth]{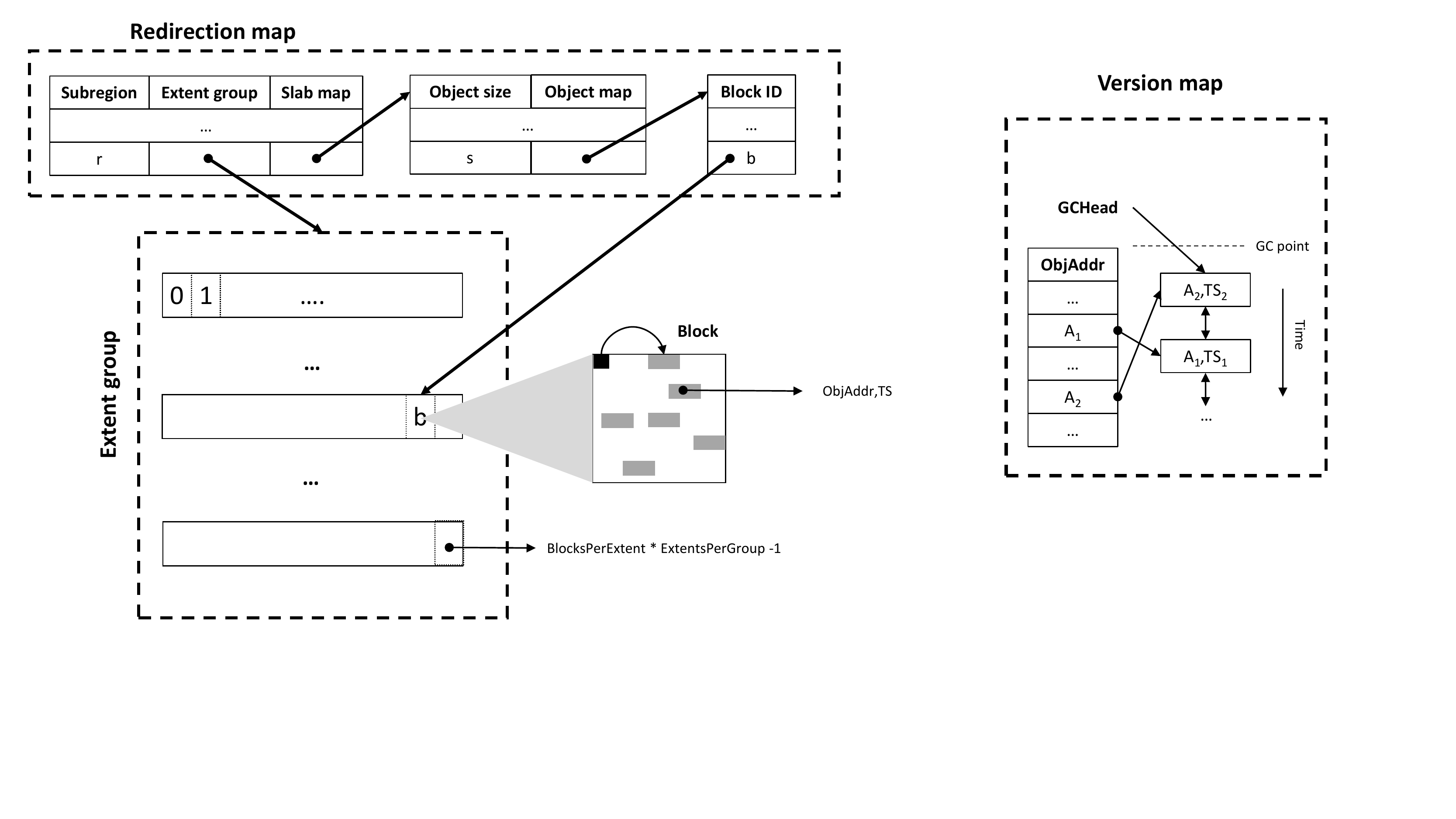}
    \caption{Redirection and version maps in \farmopaque{} with on-disk backups.}
    \label{fig:drgn_map}
\end{figure}

Figure~\ref{fig:drgn_map} shows the design of the redirection map in
\farmopaque{}. Each FaRM region is divided into subregions, 
and each subregion is further divided into fixed-size slabs. 
Each slab stores objects with a fixed size as discussed in the previous section.
The objects in a subregion are stored on a unique group of contiguous on-disk
extents, which are divided into blocks. Within each log extent, objects are written
sequentially and can span block boundaries. Each object version is written
with a header containing the object address and write timestamp. Each
block has a header that specifies the offset of the first object
header in that block. 

The redirection map is used to map an object address to the block that stores the latest version of the object.
It contains a slab map for each subregion, which is an array with an entry for each slab. 
Each entry contains the object size for the slab and an object map, which is a compact array with an entry for each object in the slab.
The entry for an object contains the identifier of the block within the extent
group for the subregion that contains the latest version of the object.

To perform an on-demand read for a given object address, FaRM finds
the block containing the start of the required object using the
redirection map and reads the block from disk. It then scans all
object headers in the block to find the matching address with the
highest timestamp version, using the block header to find the first
header and the object size for the slab (from the redirection map) to find the
position of each subsequent header.

By selecting the subregion size, extent size, extent group size, and
block size, we can trade off between disk capacity, log cleaning
overhead, redirection map size, and the cost of on-demand reads. 
For example, with 256\unit{MB} extent groups, and 4\unit{KB} blocks, each
block ID can be represented in 2 bytes, and therefore the per-object
overhead of the redirection map is 2 bytes (the additional overhead of
the per-region and per-slab tables is negligible). With 16\unit{MB}
extent groups and 64\unit{KB} blocks, the map overhead is 1 byte per
object.

When a committed transaction is applied at an on-disk backup during truncate, the
new versions of objects modified by the transaction are appended to the log and the redirection map is
updated. However, transactions are applied asynchronously and potentially out of
order at backups. Therefore, before updating the redirection map, we must
check if the version being applied has a higher timestamp than all previously logged versions for that object.

In \farmorig{}, we implemented this by storing the latest version of
each object with the block IDs in the object map. However this came at
a significant cost in memory overhead; since versions are 8 bytes in
size, the per object overhead of the redirection map was 9--10 bytes.
In \farmopaque{}, we use global time to eliminate most of this
overhead.  \farmopaque{} keeps a separate {\em version} map that maps
object addresses to their latest versions, i.e. write timestamps. The
entries in this map are also kept sorted in version order, and entries
older than the safe GC point {\em GC} are discarded, as \farmopaque{}
guarantees that we will never see an update with a timestamp older
than {\em GC}.

In practice, this eliminates almost all of the overhead of storing
versions because most objects have only an entry in the redirection
map and not in the version map. Thus an additional benefit of
implementing global time and keeping track of {\em OAT} and {\em GC}
was a 5--9x reduction in the memory overhead for on-disk backups,
which is particularly valuable with small objects (objects in FaRM can
be as small as 64 bytes).

\section{Evaluation}

In this section, we measure the throughput and latency of \farmopaque{}
and compare it to a baseline system without opacity. We then measure
the costs and benefits of multi-versioning. Finally, we demonstrate
\farmopaque{}'s scalability and high availability.

\subsection{Setup}
Our experimental testbed consists of up to $90$ machines. As we did
not always have access to all $90$ machines, all experiments used $57$
machines unless stated otherwise. Each machine has 256 GB of DRAM and
two 8-core Intel E5-2650 CPUs (with hyper-threading enabled) running
Windows Server 2016 R2. FaRM runs in one process per machine
with one OS thread per hardware thread (hyperthread). We use 15 cores
for the foreground work and 1 core for lease management and clock
synchronization. By default clocks are synchronized at an aggregate
rate of 200,000 synchronizations per second, divided evenly across all
non CMs. Each machine has one Mellanox ConnectX-3 56 Gbps Infiniband
NIC, connected to a single Mellanox SX6512 switch with full bisection
bandwidth. All graphs show average values across 5 runs with error
bars showing the minimum and maximum across the 5 runs.  Most
experiments ran for 60 seconds after a warmup period, but the
availability experiments ran for 10 minutes.

We use two benchmarks: the first is TPC-C~\cite{tpcc}, a well-known
database benchmark with transactions that access hundreds of rows. Our
implementation uses a schema with 16 indexes. Twelve of these only
require point queries and updates and are implemented as hash
tables. Four of the indexes also require range queries and are
implemented as B-Trees. We load the database with 240 warehouses per
server, scaling the database size with the cluster size. We partition
all tables by warehouse except for the small, read-only $ITEM$ table
which is replicated to all servers. We run the full TPC-C transaction
mix and report throughput as the number of neworder transactions
committed per second.

Our second benchmark is based on YCSB~\cite{ycsb}.  
We used a database of 285 million keys, with
16-byte keys and 1\unit{KB} values, stored in a single B-Tree with
leaves spread randomly across a 57-machine cluster, i.e., without any
range partitioning. The B-Tree leaves were large enough to hold
exactly one key-value pair, so a single key read or update caused one
\farmopaque{} object read or write.

We evaluate two systems. The first, \baseline{}, is an optimized
version of \farmorig{}~\cite{farm_sosp, farm_nsdi} that significantly
outperforms our previously reported TPC-C performance for \farmorig{}.
The second system is \opaque{}, which adds opacity and
multi-versioning to this optimized baseline. Both systems are run with
strict serializability and 3-way replication unless stated otherwise.
When running TPC-C we use the single-version mode of \opaque{} by
default as this gives the best performance for this workload.

\subsection{Overhead of opacity}

\begin{figure}
    \centering
    \figtrimvals
    \includegraphics[trim=0 \graphbottrim{} 0 0,clip,width=\textwidth]{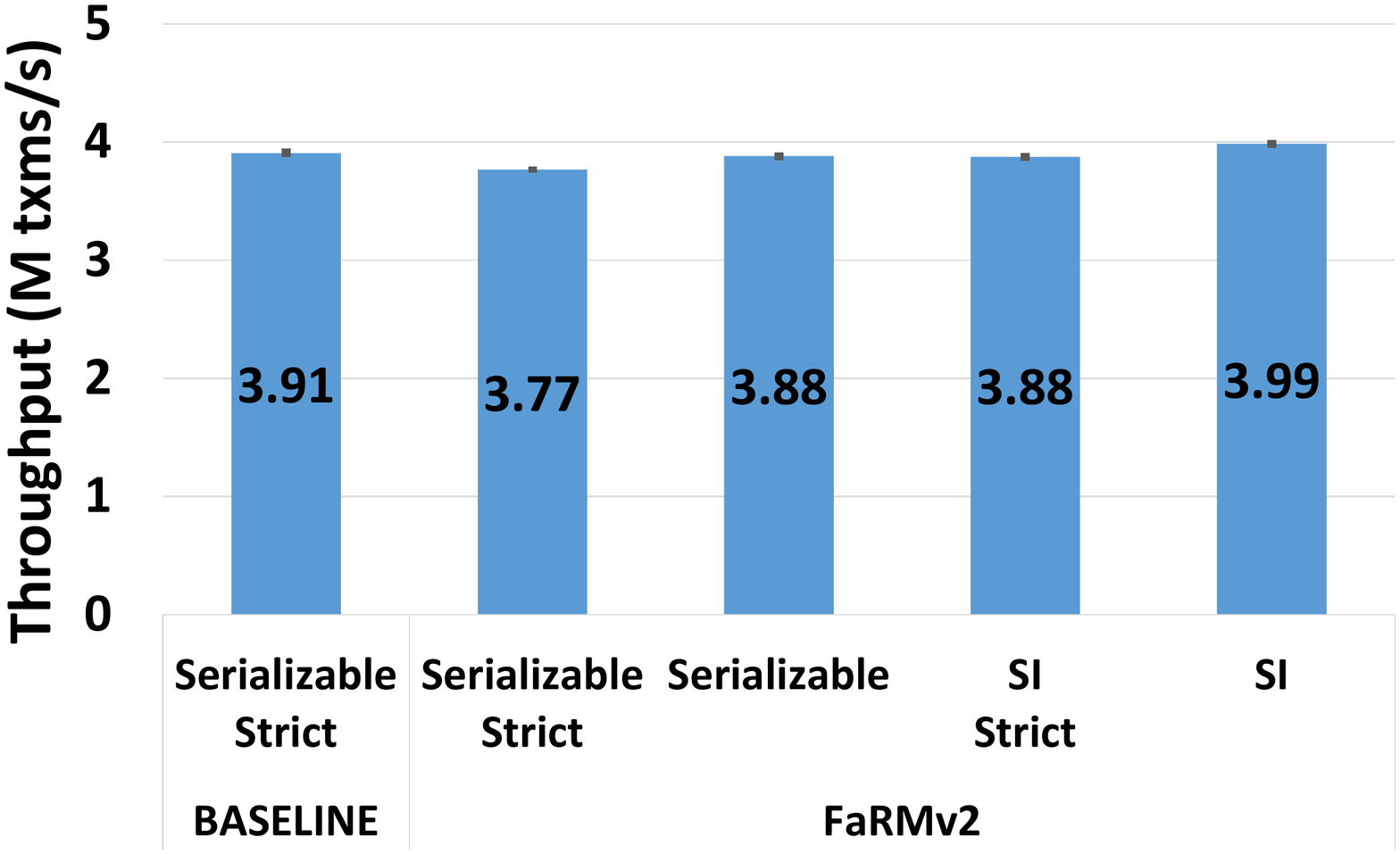}
    \caption{TPC-C throughput}
    \label{fig:tpcc-opacity}
\end{figure}

Figure~\ref{fig:tpcc-opacity} shows the saturation throughput of
\baseline{}, and of \opaque{} in single-version mode with
different combinations of serializability/SI and
strictness/non-strictness.

\opaque{} has a throughput of 3.77 million neworders per second on 57
machines with strict serializability, 3.6\% lower than \baseline{}.
The overhead of opacity comes from uncertainty waits, clock
synchronization RPCs, timestamp generation (querying the cycle counter
and thread-safe access to shared synchronization state) and $OAT$
propagation across threads and machines (which is enabled in both
single-version and multi-version mode). Abort rates are extremely low
(0.002\%) for both \baseline{} and \opaque{}.  Relaxing strictness
improves performance by 3\% with serializability, by removing the
overhead of the uncertainty wait on the read timestamp.  Using SI
rather than serializability improves performance by a further 2.7\% by
removing the uncertainty wait on the write timestamp, and also the
overhead of validation.

\begin{figure}
    \centering
    \figtrimvals
    \includegraphics[trim=0 \graphbottrim{} 0 0,clip,width=\textwidth]{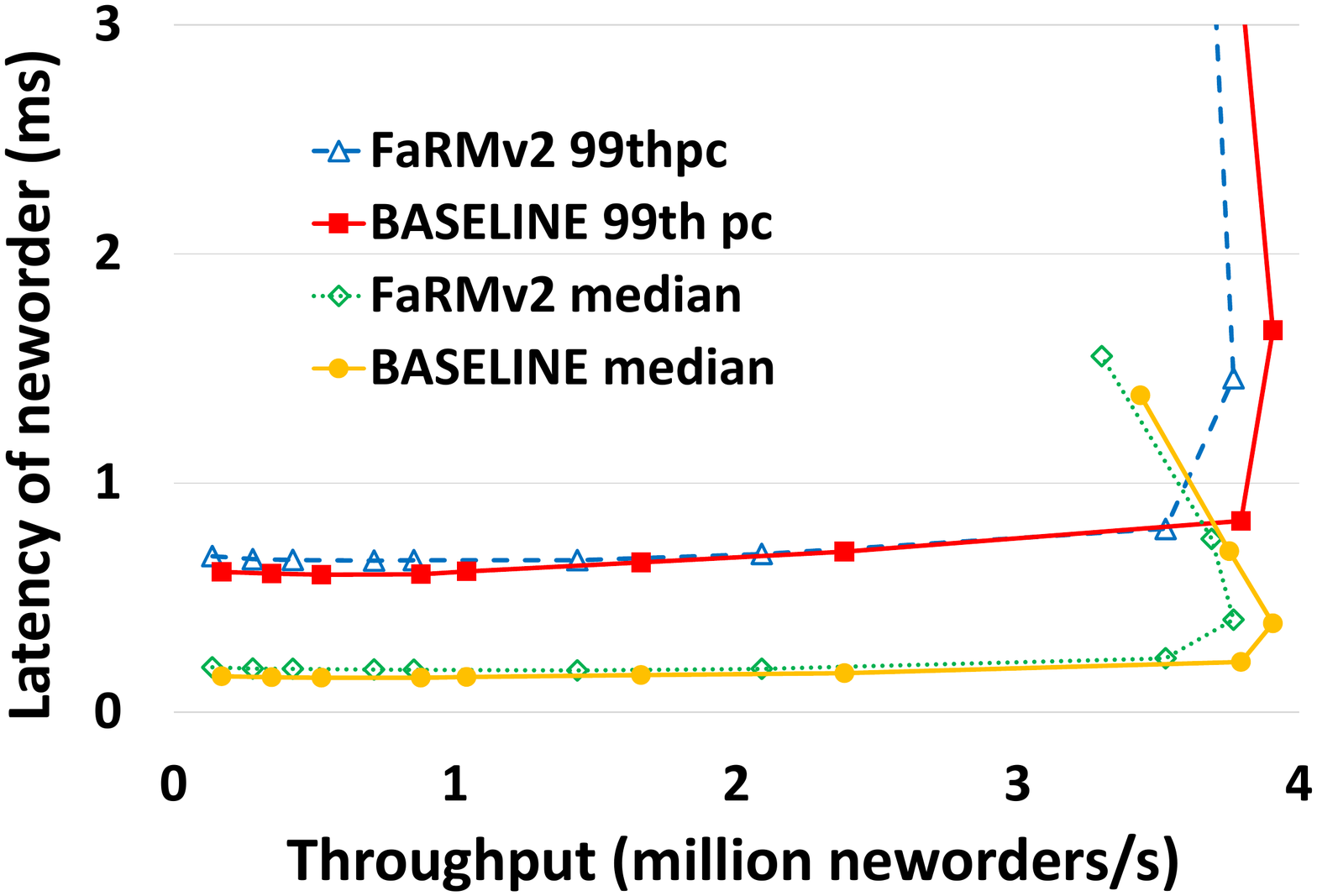}
    \caption{Throughput/latency}
    \label{fig:tpcc-latency}
\end{figure}

Figure~\ref{fig:tpcc-latency} shows a throughput-latency curve
generated by varying the level of concurrency in the workload. Both
\opaque{} and \baseline{} are able to sustain close to peak throughput
with low latency, e.g., \opaque{} has a 99th percentile neworder
latency of 835\unit{$\mu$s} at 94\% of peak throughput. At high load,
latency is dominated by queueing effects in both systems.  At low
load, the cost of opacity is an additional 69\unit{$\mu$s} (11\%)
of latency at the 99th percentile and 39\unit{$\mu$s} (25\%) at
the median.
The throughput cost of opacity is small. The added latency can be
significant for short transactions but we believe it is a price worth
paying for opacity.

\begin{figure}
    \centering
    \figtrimvals
    \includegraphics[trim=0 \graphbottrim{} 0 0,clip,width=\textwidth]{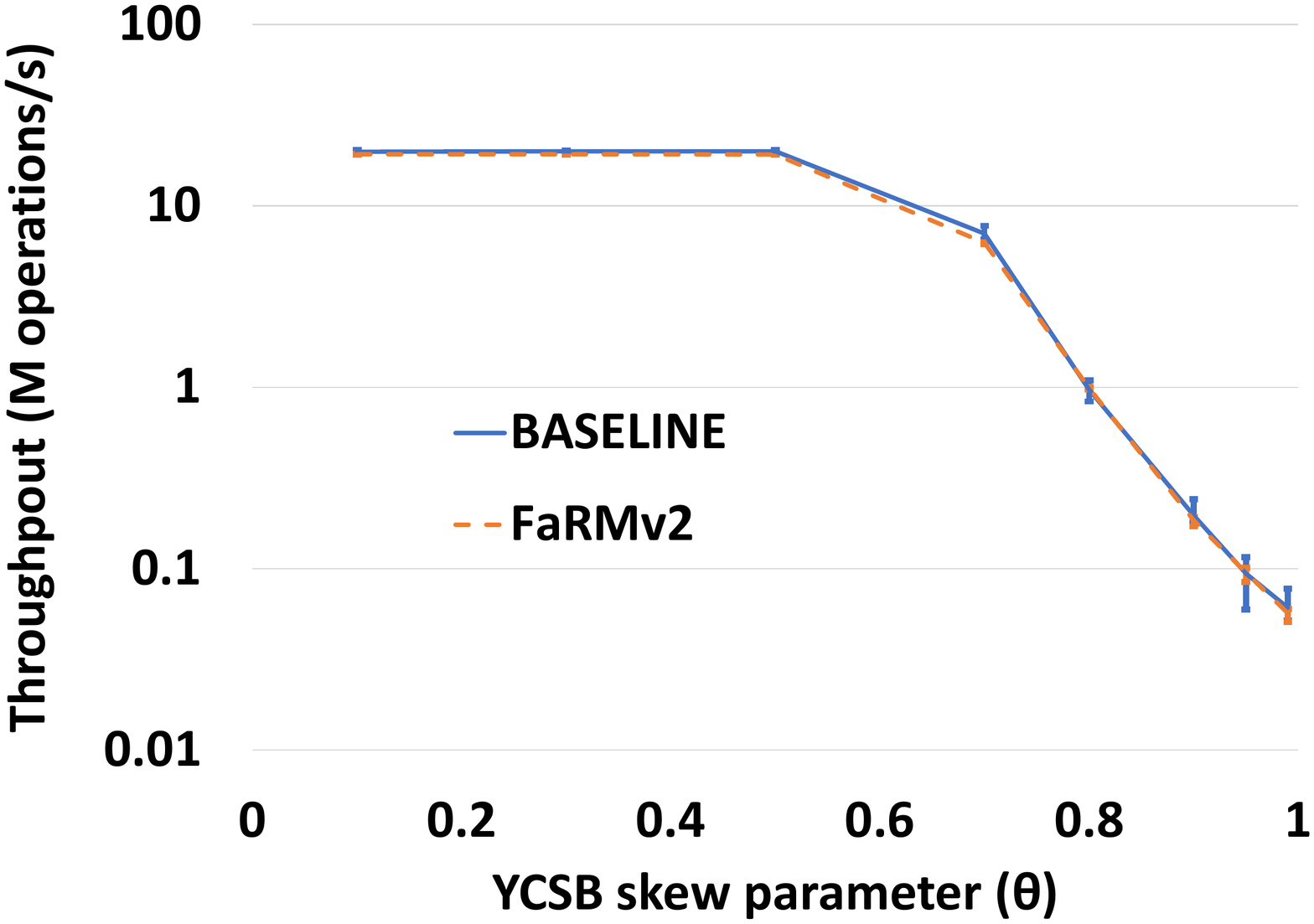}
    \caption{YCSB throughput with skew}
    \label{fig:ycsb-skew}
\end{figure}

We also measured the effect of skew on performance, using the YCSB
benchmark with a 50-50 ratio of key lookups and updates
(Figure~\ref{fig:ycsb-skew}).  Keys are selected from a Zipf
distribution, with a higher skew parameter $\theta$ giving a more
skewed distribution. At low skew, \opaque{} has around 3\% lower
performance than \baseline{}. At high skew, both systems have two
orders of magnitude lower throughput because of aborts due to high conflict rates
and performance differences are within experimental error. 
The cost of opacity does not increase with conflict rates.

\subsection{Multi-versioning}

Multi-versioning can improve performance by avoiding aborts of
read-only transactions, at the cost of allocating and copying old
versions. For TPC-C, this cost is a 3.2\% reduction in peak throughput
(1.5\% for allocation/GC and 1.7\% for copying). There is no
performance benefit for TPC-C from using MVCC as the abort rate is
extremely low even in ``single-version'' mode.

Multi-versioning has benefits when single-version operation would
frequently abort read-only transactions due to conflicts with
updates. \opaque{} bounds the memory usage of old versions to keep
sufficient space for head versions and their replicas. When this limit
is reached, writers will not be able to allocate old versions during
the $LOCK$ phase. We implemented three strategies for handling this
case, which does not occur in the TPC-C experiments shown so far. We
can {\em block} writers at the lock phase until old version memory
becomes available; we can {\em abort} them when old version memory is
not immediately available; or we can allow them to continue by allowing
old version allocation to fail, and {\em truncating} the entire history of 
objects for which old versions could not be allocated. The last option
improves write performance at the risk of aborting readers.

\begin{figure}
    \centering
    \figtrimvals
    \includegraphics[trim=0 \graphbottrim{} 0 0,clip,width=\textwidth]{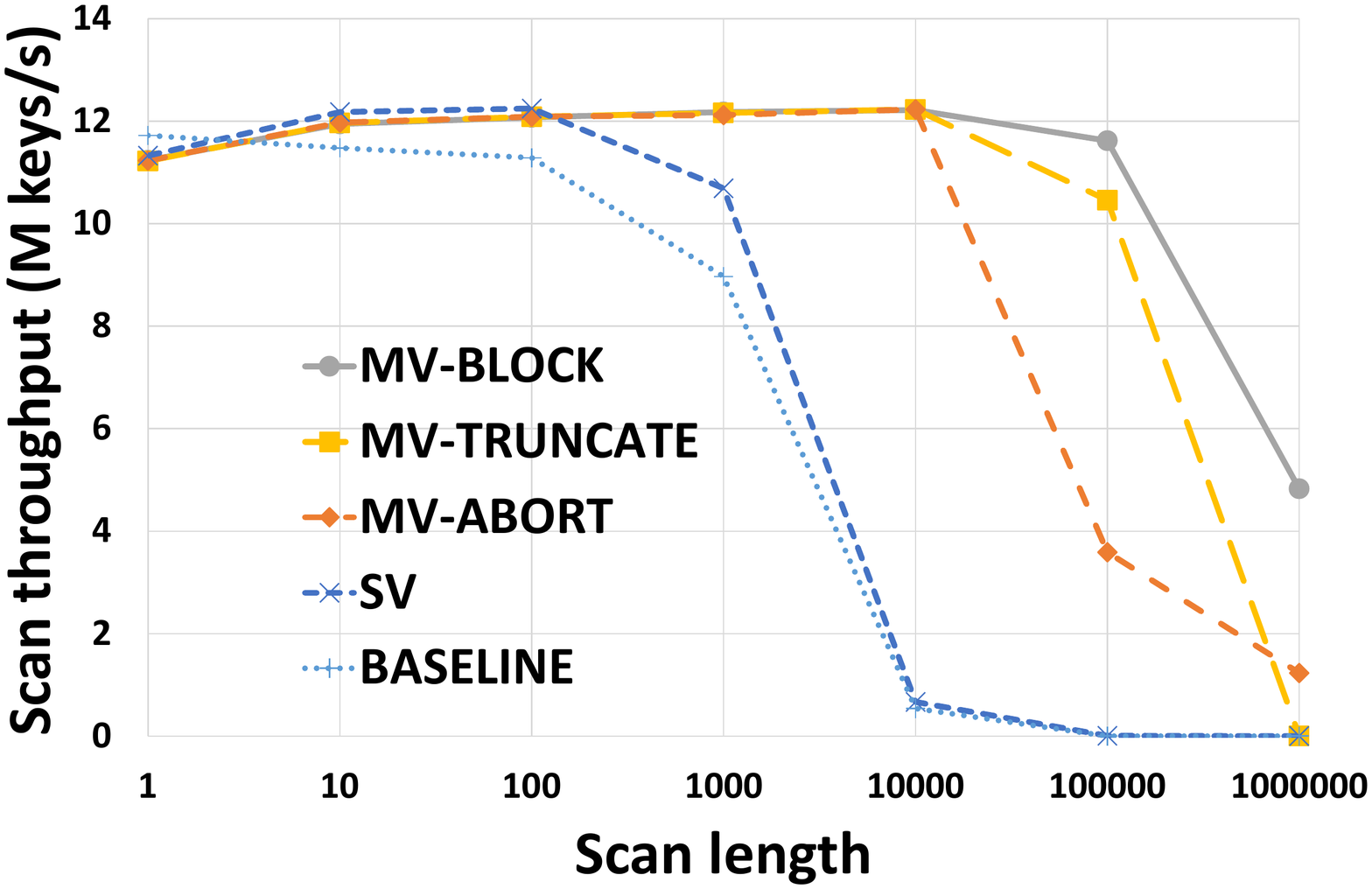}
    \caption{Throughput of mixed scan/update workload}
    \label{fig:ycsb-scans}
\end{figure}

We measured the performance of these different approaches using
YCSB. The workload contained scans (range queries) of varying length
as well as single-key updates. The start of each scan and the key for
each update is chosen from a uniform random distribution.  We
maintained a 50:50 ratio of keys successfully scanned (keys in scans
that completed without aborting) to keys successfully
updated. Old-version memory was limited to 2\unit{GB} per server.
Transactions that aborted were retried until they committed.  We
report the average throughput from 10\unit{min} of steady-state
measurement for each experimental configuration.

Figure~\ref{fig:ycsb-scans} shows the throughput in millions of keys
scanned per second on a linear scale against the scan length on a log
scale. For scans that read a single key, \baseline{} performs slightly
better than \opaque{} as it does not have the overhead of maintaining
opacity and transactions that read a single object do not perform
validation.  For scans of multiple objects, \baseline{} must validate
every object read and performs worse than \opaque{}. Beyond scans of
length 100, both single-version \opaque{} ({\sc sv}) and \baseline{}
lose throughput rapidly because of aborted scans.  The three
multi-versioning variants of \opaque{} maintain high throughput up to
scans of length 10,000: at this point, {\sc mv-abort} begins to lose
performance. {\sc mv-block} and {\sc mv-truncate} maintain throughput
up to scans of length 100,000 and then lose performance, with {\sc
  mv-block} performing slightly better than {\sc mv-truncate}. For
longer scans, all strategies perform poorly.

The average scan latency for 10,000 keys was 750--850\unit{ms} and all
the {\sc mv} variants perform well at this point.  Our
target applications have lower update rates and shorter queries. 
All the {\sc mv} variants perform well for these
relatively short queries, without any drop in throughput. In production
we use {\sc mv-truncate} by default.

\subsection{Scalability}

\begin{figure}
    \centering \figtrimvals \includegraphics[trim=0 \graphbottrim{} 0
      0,clip,width=\textwidth]{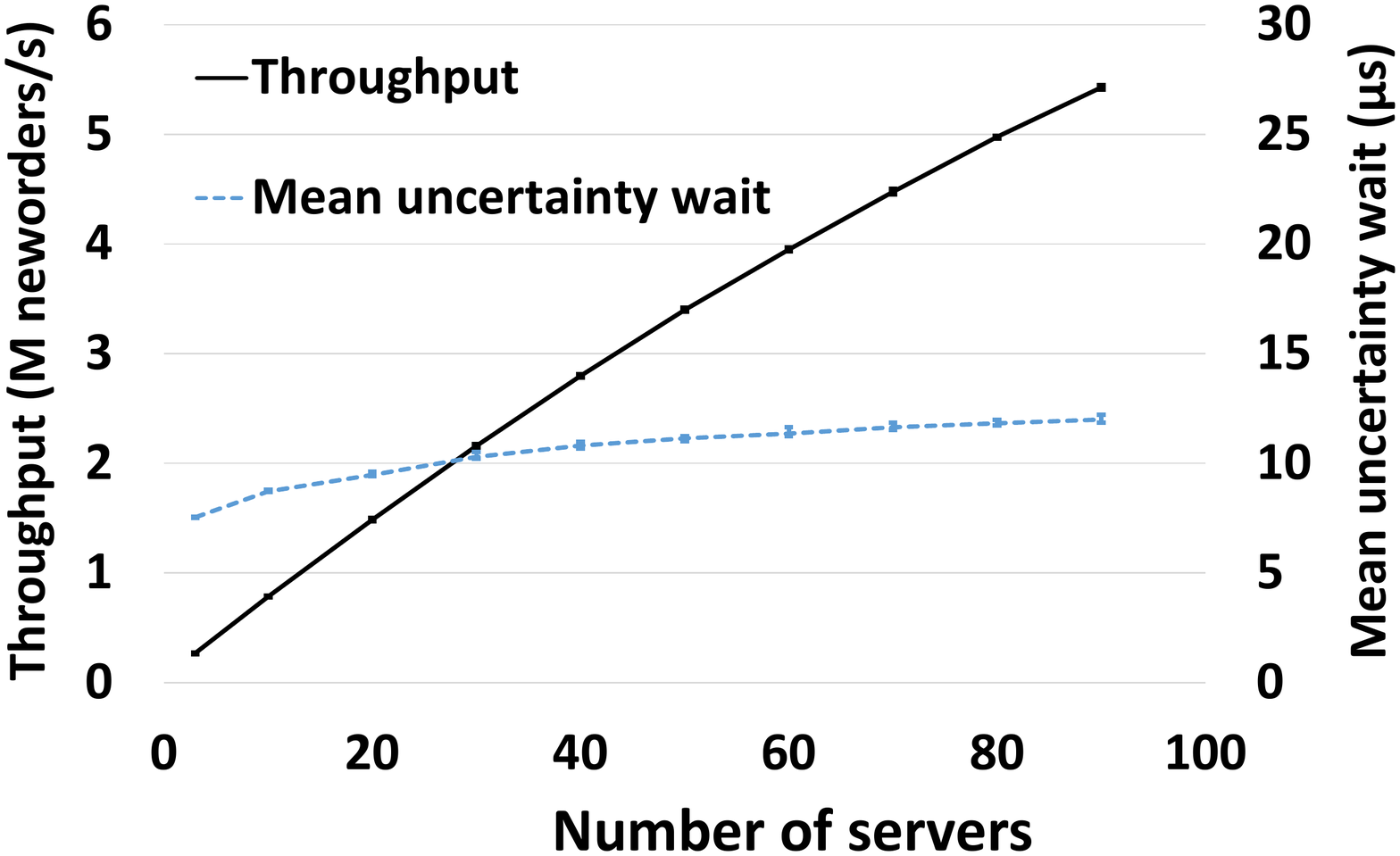}
    \caption{Scalability of \opaque{}}
    \label{fig:tpcc-scalability}
\end{figure}

\begin{figure}
    \centering
    \figtrimvals
    \includegraphics[trim=0 \graphbottrim{} 0 0,clip,width=\textwidth]{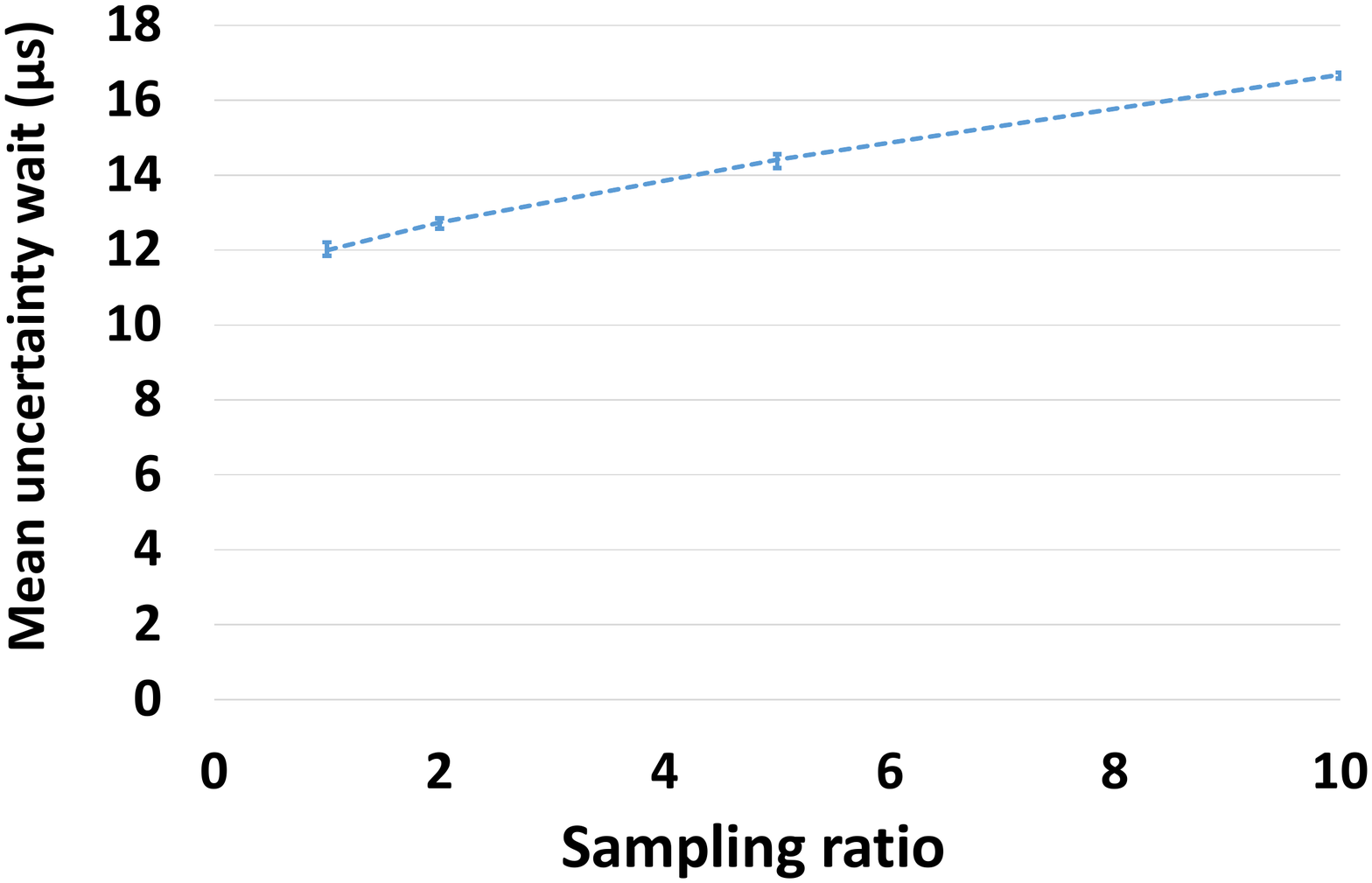}
    \caption{Scalability of clock synchronization}
    \label{fig:tpcc-clocksync}
\end{figure}

Figure~\ref{fig:tpcc-scalability} shows the throughput of \opaque{}
with strict serializability as we vary the cluster size from 3, the
minimum possible with 3-way replication, to 90. Throughput scales
well, achieving 5.4 million neworders/s at 90 machines with 21,600
warehouses. This is the highest TPC-C transactional throughput we know
of for a system with strict serializability and high
availability.\footnote{Like those of other in-memory
  systems~\cite{farm_sosp,silo_osdi14,drtm_eurosys}, our results do
  not satisfy the TPC-C scaling rules for number of warehouses.}

We use an aggregate synchronization rate of 200,000 synchronizations
per second at the CM. As the cluster size increases, this means
individual non CMs synchronize less frequently, causing the
uncertainty to increase. Figure~\ref{fig:tpcc-scalability} also shows
the average uncertainty wait as a function of cluster size, which
increases only moderately from 7.5\unit{$\mu$s} to 12\unit{$\mu$s}: a
60\% increase in uncertainty for a 30x increase in cluster size.

We emulated the effect of clusters larger than 90 machines on
uncertainty by down-sampling the synchronization responses from the
CM, on each non CM. For example, at a sampling ratio of 10 with 90 machines,
non-CMs discard 9 out of 10 synchronization responses, which emulates
the rate at which non-CMs would synchronize in a cluster of 900
machines. Figure~\ref{fig:tpcc-clocksync} shows the effect of
down-sampling on a 90-machine cluster with a sampling ratio varying
from 1 to 10. Across this factor of 10, the mean uncertainty wait
increases by 39\%, from 12\unit{$\mu$s} to 16.7\unit{$\mu$s}.

Thus our centralized time synchronization mechanism scales well
with only small increases in uncertainty when the cluster
size is increased by orders of magnitude. While other factors such
as network bandwidth and latency might limit application performance,
the synchronization mechanism scales well to moderate sized clusters
of up to 1000 machines.

\subsection{Availability}

\begin{table*}
\begin{center}
\begin{tabular}{|l|rr|rr|r|}
\hline
\multicolumn{1}{|c|}{Machines failed} &  \multicolumn{2}{|c|}{Clock disable time} & \multicolumn{2}{|c|}{Recovery time} & Re-replication time \\
\hline
1 non-CM        &  0             &                 & 49\unit{ms} & (44--56\unit{ms})  & 340\unit{s} (336--344\unit{s}) \\
CM              &  4\unit{ms}  & (3--4\unit{ms})  & 58\unit{ms} & (27--110\unit{ms}) & 271\unit{s} (221--344\unit{s})  \\
CM and 1 non-CM &  16\unit{ms} & (11--20\unit{ms}) & 71\unit{ms} & (61--85\unit{ms})  & 292\unit{s} (263--336\unit{s})\\
\hline
\end{tabular}
\end{center}
\vspace{1cm}
\caption{Recovery statistics}
\label{tab:failure}
\end{table*}

We measured availability in \opaque{} by killing one or more
\opaque{} processes in a 57-machine cluster running the TPC-C
workload and using 10\unit{ms} leases for failure detection. For these experiments we measure throughput over time in
1\unit{ms} intervals. Table~\ref{tab:failure} summarizes results for 3
failure cases: killing a non-CM, killing the CM, and killing both the
CM and a non-CM. It shows mean values across 5 runs with the min/max
range in parentheses.

The clock disable time is the elapsed time between the new CM
disabling and re-enabling the clock. The recovery time is the time
from when a failure was first suspected, to the time when the
throughput after failure reaches the mean throughput before the
failure. The re-replication time is the time taken for \opaque{}
to bring all regions affected by the failure back to full (3-way)
redundancy.

\begin{figure}
    \centering
    \figtrimvals
    \includegraphics[trim=0 \graphbottrim{} 0 0,clip,width=\textwidth]{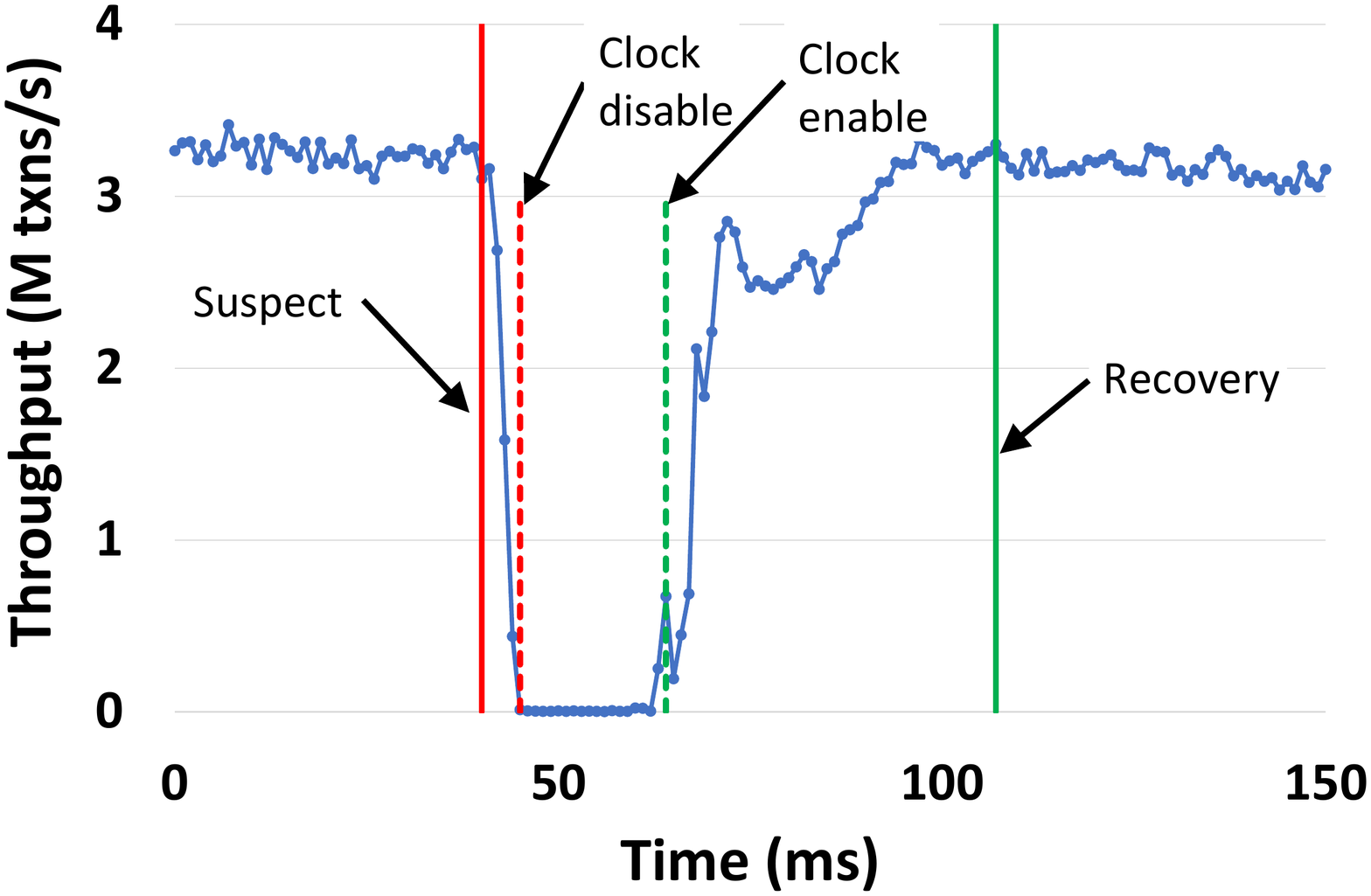}
    \caption{Failure and recovery timeline}
    \label{fig:tpcc-timeline}
\end{figure}

Figure~\ref{fig:tpcc-timeline} shows throughput over time for one of
the runs where we killed the CM as well as a non CM.  The solid
vertical lines mark the point at which \opaque{} first suspects a
failure, and the point at which we recover full throughput.  The
dashed vertical lines show the points at which the CM's clock was
disabled and then enabled. The timeline does not show the
re-replication time, which is much longer than the timescale of the
graph.

\opaque{} provides high availability with full throughput regained
within tens of milliseconds in most cases. Even before this point,
transactions can continue to execute if they only read objects whose
primaries did not change, they only write objects none of whose
replicas changed, and clocks are enabled. When clocks are disabled,
transactions block when trying to acquire read or write
timestamps. Clocks are not disabled if only non-CMs fail. If only the
CM fails they are disabled for 4\unit{ms} on average, and if both the CM and a
non-CM fail they are disabled for 16\unit{ms} on average.

Data re-replication takes around 5\unit{minutes} because it is paced
to minimize the performance impact on the foreground workload.
Re-replication happens in the background and concurrently with the
foreground workload. It does not affect availability. The pacing is a
configurable parameter: more aggressive re-replication will shorten
the window of vulnerability to additional failures but will contend
for network and CPU resources with the foreground workload.

\subsection{Operation logging}

NAM-DB~\cite{namdb} uses replicated in-memory operation logging with
no checkpointing or logging to disk. Data is not replicated; instead
each committed RW transaction writes the transaction description,
inputs, and write timestamp to in-memory logs on three machines. The
logging is load-balanced across all machines. This configuration has
high TPC-C throughput (6.5 million neworders/s on 57 machines) but it
is not realistic: once the in-memory logs fill up, the system cannot
accept any more writes. A realistic configuration would need frequent
checkpointing to truncate the logs, with substantial overhead.

Operation logging also hurts availability: on any failure, all the
logs must be collected in a single location, sorted, and the
operations re-executed sequentially, which can be orders of magnitude
slower than the original concurrent execution. 

\opaque{} provides high availability through replication. It uses
in-memory logs at backups which are continuously truncated by applying
them to in-memory object replicas.  During recovery, after clocks have
been enabled, transactions can read any region whose primary did not
change and write to a region if none if its replicas changed. Other
regions can be accessed after a short lock recovery phase during which
the untruncated portion of the in-memory logs are scanned in parallel
by all threads and write-set objects are locked.

For a fair comparison we implemented operation logging in \opaque{}
and measured the same configuration as NAM-DB: operation logging,
multi-versioning, non-strict snapshot isolation.  \opaque{} achieves
9.9 million neworders/s with 90 machines. With 57 machines, the
throughput is 6.7 million neworders/s, 4\% higher than NAM-DB's
throughput with the same number of machines, comparable CPUs, and
using older, slower CX3 NICs.

\section{A1 - Distributed Graph Database}
\label{sec:a1}

A1~\cite{A1_SIGMOD, a1_hpts} is a scalable, transactional, distributed graph database which was built on top of FaRM. 
A1 and FaRM are used in production as part of Microsoft's Bing search engine.

A1 stores large property graphs that evolve in real time, i.e., both edges and vertices have associated data (properties) and both the structure of the graph and the properties can change in real time.
A1 applications require serving complex queries within a tight latency budget and with high availability. Since these queries may involve accessing a large number of vertices and edges with long paths of data dependent accesses, low access latency is an important requirement. High throughput is also a requirement to achieve low cost graph serving. FaRM can meet all these requirements by scaling out in-memory storage and using new algorithms to leverage fast networking and non-volatile memory for low latency, high throughput, and high availability. The developers of A1 chose to use FaRM not only because it met all their requirements but also because it provided a simple programming model abstraction of a single large machine that runs transactions sequentially and never fails. No other platform met their requirements.

 \begin{figure}[t]
      \centering
      \figtrimvals
      \includegraphics[trim=0 0 0 0,clip,width=\textwidth]{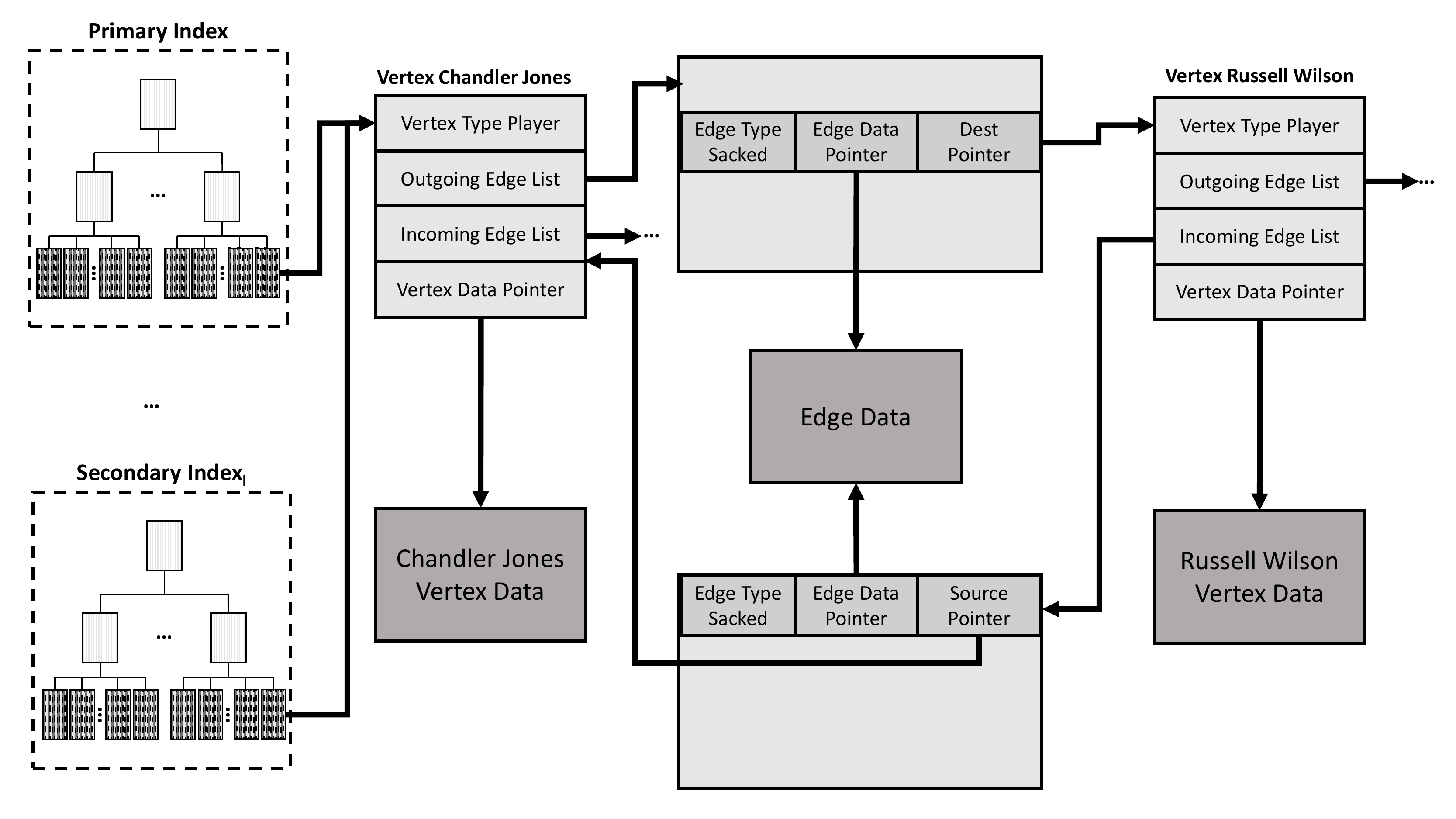}
      \begin{captiontext}
        A1 uses a collection of FaRM objects, pointers, and B-trees to represent a property graph.
        The solid lines represent FaRM pointers and the representation is optimized to traverse edges in both directions.
      \end{captiontext}
      \caption{Representation of property graphs in FaRM}
      \label{fig:A1_data_layout}
  \end{figure}
  
A1 represents the property graph using FaRM objects, pointers, and B-trees.
\Cref{fig:A1_data_layout} shows an example for two vertices representing NFL players Chandler Jones and Russell Wilson, 
and an edge indicating that Chandler Jones sacked Russell Wilson. The solid black arrows in the figure are FaRM pointers.
Vertices and edges are represented using several FaRM objects:
  \begin{itemize}
      
    \item \textbf{Vertex Header -} This object contains a field that identifies the vertex type.
    The type is represented by another object that includes a schema describing the data associated 
    with the vertex and its primary key. It also contains pointers to three other objects: a list of outgoing edges, 
    a list of incoming edges, and the data 
    (properties) associated with this vertex. As a space saving optimization, it is possible to embed some or all 
    of these objects in the vertex header when they are small.
     
    \item \textbf{Outgoing Edge List -} This object is an array with an entry for each edge from this vertex to other vertices. 
    The entry for an edge includes the edge type, a pointer to the data associated with the edge, and a pointer to the destination vertex. 
    The type is represented by another object that includes a schema describing the data associated with the edge. If the number of outgoing edges is larger than a threshold, the edge list is stored in a FaRM B-tree.
    
    \item \textbf{Incoming Edge List -} This object is an array with an entry for each edge from other vertices to this vertex. 
    It enables queries that traverse the edge in reverse. The entry for an edge includes the edge type, a pointer to the data 
    associated with the edge, and a pointer to the source vertex. 
    As with outgoing edges, if the number of incoming edges is larger than a threshold, the edge list is stored in a FaRM B-tree.

    \item \textbf{Vertex Data -} This object contains the data associated with the vertex. 
    For example, the vertex data for Russell Wilson may include the player's name, height, and date and place of birth. 
    
    \item \textbf{Edge Data -} This object contains the data associated with the edge. 
    It is pointed to by the edge headers in both the outgoing and incoming edge lists. 
    For example, the edge data in the example could include a particular date and time 
    when Chandler Jones sacked Russell Wilson. There would be a different edge for each time Chandler Jones sacked Russell Wilson.
    
  \end{itemize}

  In addition to these objects, A1 leverages FaRM B-trees to index data associated with vertices and edges.
  A1 requires each vertex type to have a primary index. Vertex types can also have one or more secondary indices. Edge types do not require a primary index.
  
A1 leverages FaRM's distributed transactions extensively to simplify development and also to expose transactions to users. Transactions are useful even for simple CRUD operations like adding an edge between two vertices. This operation involves adding edge headers to the appropriate outgoing and incoming edge lists and allocating and initializing the edge data object. These objects may be stored in different servers. A1 performs all these operations in a FaRM transaction without worrying about distribution, failures, or concurrent accesses. Development is harder in other systems that provide transactions only for operations that do not cross server boundaries. For example, Facebook's Tao~\cite{tao} requires a scrubbing mechanism to fix edges that were installed only in one direction. A1 also leverages FaRM's distributed transactions for maintaining distributed global indices consistently and for schema evolution.

Initially, A1 was developed on top of \farmorig{}. It was the experience of A1 developers that led us to add opacity in \farmopaque{}. Without opacity, A1 developers had to program defensively because they could not rely on invariants during transaction execution (because the transaction could later abort). For example, all the pointers in~\Cref{fig:A1_data_layout} were fat pointers that included both an address and an incarnation number, which had to match the incarnation number in the object. All the issues discussed in Section~\ref{sec:motiv} were brought up by A1 developers. Since the introduction of opacity, A1 developers report a better development experience with \farmopaque{}. Additionally, we were able to save space and complexity by turning fat pointers into simple FaRM addresses and removing other code that was previously required for defensive programming.

Another feature of \farmopaque{} that is important for A1 are parallel distributed read-only transactions, which were discussed in Section~\ref{sec:dist_ro}. These transactions allow A1 to parallelize large queries across many servers to reduce latency without compromising consistency or performance. Parallel distributed read-only transactions are strictly serializable and they leverage MVCC for avoiding aborts. A1 uses this feature as follows. A server that receives a complex query becomes the master. The master starts a read-only FaRM transaction and retrieves its read timestamp (rts). It then starts the query traversal and identifies a set of servers with objects of interest. It sends RPCs to those servers intructing them to continue the traversal from those objects while performing all reads at rts to ensure the query executes at a consistent snapshot. 
This process of fanning out read-only transactions and collecting their
  results may be repeated multiple times, depending on the complexity of the
  query. When the query is complete the master, can commit the
  transaction allowing old versions used in the snapshot to be garbage collected.
  
\section{Proof of Correctness}
\label{sec:proof-of-correctness}

We present a proof of opacity for a simplified variant of the \farmopaque{} commit protocol. 
For simplicity, our argument will assume abstract variants of replication, approximate timestamp maintenance, and object locking. The resulting simplified protocol is described below. 

\paragraph{Simplified Protocol.}
Recall that we assume the timestamping protocol in Algorithm~\ref{algo:ts}, where $\epsilon$ is a known bound on drift. Informally, this mechanism ``waits out" the uncertainty before returning a timestamp which is guaranteed to be 1) in the past at the moment when the function returns and 2) have occurred during the interval of the call. 
\begin{algorithm*}[ht]
\caption{Timestamp generation for strict serializability.}
 \label{algo:ts}
\begin{algorithmic}[1]
 
\Function{GET\_TS}{ } 

	\State ${[L, U]} \gets$ $\mathsf{TIME()}$
	\State $\mathsf{SLEEP}( (U - L) (1 + \epsilon) ) $
	\State \textbf{return} $U$
\EndFunction
\end{algorithmic}
\end{algorithm*}

Given this mechanism, the transactional commit protocol for strict serializability is described in Algorithm~\ref{algo:tx}. The pseudocode closely follows the description in the main paper body. Therefore we do not reiterate the textual description. For simplicity, the pseudocode merges transaction execution and commit in the function \texttt{ExecuteAndCommit} and it assumes the read and write sets are passed as arguments to this function to simplify the exposition. The implementation supports data dependent execution and the read and write sets are computed dynamically during transaction execution, but this does not fundamentally affect the commit protocol.

\newcommand\CommentLine[1]{$\triangleright$ \parbox[t]{\dimexpr\linewidth-\algorithmicindent}{#1}}

\begin{algorithm*}[ht]
\caption{Transactional Algorithm Pseudocode.}
 \label{algo:tx}
\begin{algorithmic}[1]
 
\Function{ReadAtTs}{R, ts}
	\If{ R is locked }
		 \textbf{return} $\mathsf{NULL}$ 
	\Else
		 \, Return version of object R that has the highest timestamp which is $\leq$ than ts or $\mathsf{NULL}$ if none available
	\EndIf
\EndFunction
\Statex
\Function{LockAtTs}{R, ts}
	\If{ R is locked or R has timestamp $\geq$ ts }
		 \textbf{return} $\mathsf{NULL}$ 
	\Else
		 \, Lock object R 	
	\EndIf
\EndFunction
\Statex
\Function{ExecuteAndCommit}{ RSet, WSet } 
	\State rts $\gets \Call{GET\_TS}$ \label{line:getrts}
	
\CommentLine{Execute reads}
	
	\For{ R in RSet }
		\If{ ! \Call{ReadAtTs}{ R, rts } } \label{line:successful-read}
			\State \texttt{ABORT} 
			\CommentLine{Object locked or old version not available.} 
		\EndIf
	\EndFor
	\If{WSet $== \emptyset$} 
	    \State \texttt{COMMIT} \CommentLine{Read-Only Transaction.} 
	\EndIf
	
\CommentLine{Read-Write Transaction}
	
	\For {W in WSet}
	
		\If{ ! \Call{LockAtTs}{ W, rts } }
			\State \texttt{ABORT} 
		\EndIf
	\EndFor
	
\CommentLine{Waits out uncertainty while holding locks}
	\State wts $\gets \Call{GET\_TS}$ \label{line:wwait}

	 \For { R in RSet $\setminus$ WSet }
	
		\If{ R is locked or R has timestamp > rts) }
		    \State \texttt{ABORT} \CommentLine{Validation failed.} 
		\EndIf
		
	\EndFor
			
	\State Apply a new version with wts to each variable in WSet
	\State Unlock all and  \texttt{COMMIT} 
	
\EndFunction
\Statex
\end{algorithmic}

\end{algorithm*}

\subsection{Preliminary Results}

We begin the proof by listing a series of simple properties and invariants.
The first states simple properties of the timestamps returned by \texttt{GET\_TS}. 
In the following, we will make the following claims about the global time, i.e., the time at the clock master, at which certain execution events occur, denoted by  $\mathcal{R}(\cdot)$. 

\begin{lemma}
\label{lem:simple}
The following hold.
\begin{enumerate}
    \item If $\mathcal{R}(time)$ is the global time at which the invocation of \texttt{TIME} returning the interval ${[L, U]}$ occurs, and $\mathcal{R}(ret)$ is the global time at which the  \texttt{GET\_TS} function returns, it holds that 
    \begin{equation}
        L \leq \mathcal{R}(time) \leq U \leq  \mathcal{R}(ret). 
    \end{equation}
    In a nutshell, the returned timestamp $U$ is in the past at the point when the function returns, and the global time $U$ occurs during the  execution of \texttt{GET\_TS}.
    
    \item Let $wts$ be the write timestamp of a given transaction, and let $\mathcal{R}(\mathsf{lock}(W))$ be the global time at which the lock at object $W$ completes successfully at the coordinator. Then, for any object $W$, $\mathcal{R}(\mathsf{lock}(W)) \leq wts$. In other words, every committed write transaction holds all locks at the global time given by its write timestamp. 

\end{enumerate}

\end{lemma}

The proof of this lemma follows immediately from the properties of the timestamp mechanism, and is therefore omitted. 
The proof of opacity in the next section will center around two invariants. 
The first concerns values read by each transaction. Intuitively, it states that for each object, the value read by a transaction is the latest version ever written with a write timestamp less than or equal to its read timestamp.

\begin{lemma}[Read Invariant]
\label{lemma:read-invariant}
	Let  $rts$ be a transaction's read timestamp, and let $R$ be an object in the transaction's read set, which is read at version $v_R \leq rts$. 
	Then $R$ may never be updated at a version in the interval $(v_R, rts)$ in this execution. 
\end{lemma}
\begin{proof}

    Let the current transaction be $T$, and assume for contradiction that there exists a transaction $T'$ which writes a version $w \in (v_R, rts)$ at some point in this execution. 
    Given that the read succeeds, it follows that the object $R$ cannot be locked at the time when it is read. 
    Since the latest version read before $rts$ is $v_R$, this version could only have been written \emph{after} the transaction $T$ completed its read in line~\ref{line:successful-read}. Since the object $R$ is not locked at the time when it is read by $T$, it follows that the object may only be locked \emph{after}  $T$ completes its read of $R$. Hence, by the first statement in Lemma~\ref{lem:simple}, we have $rts \leq \mathcal{R}(read(T, R)) < \mathcal{R}( lock(T', R))$. Yet,  $\mathcal{R}( lock(T', R)) \leq w$, by the second part of Lemma~\ref{lem:simple}, as $w$ is the write timestamp of transaction $T'$. We have obtained that $w > rts$, a contradiction. This establishes our claim. 
\end{proof}

The serializability proof will follow from the following \emph{write invariant}, which says that for any committed read-write transaction, none of the objects in the union of its read and write sets may be modified in the interval between its read and write timestamps. 

\begin{lemma}[Write Invariant]
	\label{lem:write-invariant}
	Consider an arbitrary \emph{committed} read-write transaction $T$, 
	 let $rts$ be its read timestamp, and $wts$ be its write timestamp. 
	Then none of the variables in $T$'s data set (i.e. read set $\cup$ write set) may ever have a version $v \in ( rts, wts ]$. 
\end{lemma}
\begin{proof}
	Assume for contradiction that there exists an execution with a committed read-write transaction $T$ with read timestamp $rts$, write timestamp $wts$, and a variable $V$ in $T$'s data set which at some point in time has a version $v \in ( rts, wts ]$. 
	Then clearly $V$ must have been written by another committed read-write transaction $T_v$. 
	
	Lemma~\ref{lem:simple} (part 2), $T_v$ must hold the lock on $V$ at the global time $v  \in  ( rts, wts ]$. 
	We consider three cases, depending on the relation between $T_v$'s lock acquisition and $T$'s validation. First, if $T_v$ releases its lock on $V$ before $T$'s validation on $V$, then transaction $T$ will necessarily abort, as it attempts to validate $V$ at $rts$, while its latest committed version is $v > rts$. 
	
	In the second case, if $T_v$ still holds the lock on $V$ at the time of $T$'s validation, then this validation clearly fails, and $T$ aborts. The last remaining case is if $T_v$ acquires its lock on $V$ \emph{after} $T$'s validation on $V$. However, in this case $v > wts$, a contradiction.  
	
	We therefore conclude that no object in $T$'s data set may have versions between $rts$ and $wts$, as claimed. 
\end{proof}

\subsection{Proof of Opacity}

\noindent We now extend these invariants to show that the transactional protocol provides \emph{strict snapshot reads}, that is, the read set viewed by a transaction corresponds to an atomic snapshot taken during the transaction's execution. 

\begin{lemma}[Strict Snapshot Reads]
\label{lem:si}
	The transactional protocol ensures that each transaction works on a consistent snapshot of object values, taken at its read timestamp, which occurs during the transaction's execution. 
\end{lemma}
\begin{proof}
	Consider an arbitrary transaction $T$, reading a set of objects $R_1, R_2, \ldots, R_k$. 
	Applying  Lemma~\ref{lemma:read-invariant}, we have that, for each object $R_i$, the version read $v_i$ is the latest version written \emph{before} the read timestamp $rts$ of $T$. We can therefore serialize $T$'s snapshot at $rts$. Lemma~\ref{lem:simple} (part 1) guarantees strictness. 
\end{proof}

We now complete the proof by showing that the \farmopaque{} protocol satisfies strict serializability. 

\begin{lemma}[Strict Serializability]
    The \farmopaque{} protocol is strictly serializable. Read transactions are serialized at their read timestamp (rts), while write transactions are serialized at their write timestamp (wts). 
\end{lemma}
\begin{proof}
    By Lemma~\ref{lem:si}, we know that read transactions can be (strictly) serialized at their $rts$. We now show that committed read-write transactions can be serialized at $wts$. 
    
	By Lemma~\ref{lem:si}, we know that read-write transactions work on a consistent snapshot of variables, performed at $rts$.  
	Lemma~\ref{lem:write-invariant} shows that, for every object accessed by the transaction, no versions may be written in the interval $(rts, wts]$. Alternatively, no committed transaction writing to objects accessed by $T$ may be serialized in the interval $(rts, wts]$. 
	This covers possible conflicts between the execution intervals of read-write transactions, and concludes the proof of strict serializability for \farmopaque{}.
\end{proof}

\subsection{(Not) Waiting Out Uncertainty}

It is tempting to ask whether it is possible to skip the waiting interval on the $\mathsf{GET\_TS}$ call while the locks are taken. 
Unfortunately, in this case serializability may be broken, as we show by the following counterexample. 

Assume we have two variables $A$ and $B$, both initially set to 0. We will instantiate four transactions. 
T1 executes on a machine with high uncertainty. T2, T3, T4 all execute on machines with low uncertainty. The transactions are precisely defined as follows:

\begin{enumerate}
    \item T1 executes $A \gets B + 1$.

    \item T2 executes $B \gets 1$.

    \item T3 reads $A$ and $B$.

    \item T4 reads $A$ and $B$.

\end{enumerate}

For simplicity, we will consider transactions T2, T3, T4 each executing in a single time step atomically, which is also their read timestamp, and for T2 the write timestamp. We now describe the execution at each global time step. The notation $X@t$ means we are object $X$ with version $t$. 

Initially, $A@0 = B@0 = 0$.

\begin{enumerate}

\item            T1 acquires read timestamp $1$, reads $B@0 = 0$.

\item            T1 locks $A$.

\item            T1 acquires write timestamp $9$. 

\item            T1 validates $B@0$.

\item            T1 writes $A@9 \gets 1$. 

\item            T1 unlocks $A$.

\item            T2 executes with read and write timestamp $7$, writes $B@7 \gets 1$.

\item            T3 executes with read timestamp $8$, reads $A@0 = 0, B@7 = 1$.

\item            Now T1's write set i.e. $A@9 = 1$ becomes visible.

\item            T4 executes with read timestamp $10$, reads $A@9 = 1, B@7 = 1$. 
\end{enumerate}

This execution is not serializable. T4 sees a state that is only possible if T1 serializes before T2, but T3 sees a state that is only possible if
T2 has committed but T1 has not yet committed. The problem is that T1's write set is unlocked while its write timestamp is still in the future. If
it had remained locked until the write timestamp was in the past, then T3 would abort and we would have a serializable execution with serialization
order (T1, T2, T4).

\paragraph{Discussion.} The above counterexample shows that serializability is violated if write set locks are released before the write timestamp
is known to be in the past. We can extend this counterexample in the following ways:

\begin{itemize}
	\item It is not correct to perform the pre-commit validation in parallel with the waiting out of the uncertainty. This can break the write invariant, which in turn implies a breaking of serializability by a modified version of this counterexample.
	
	\item The transactional protocol which reads in the past and writes without waiting out uncertainty is not correct. This breaks the property that locks are held at the global time corresponding to the write timestamp, which leads to a similar counterexample, breaking serializability. 
	
	\item The variant of the protocol which waits out uncertainty \emph{before taking locks on the write set} is not correct. This can lead to locks not being held at the write timestamp, which can be used to break serializability, again by a similar counterexample. 

\end{itemize}

\section{Related work}

In this section we compare \farmopaque{}'s design with some other
systems that support distributed transactions. We do not aim to cover
the entire space of distributed transaction protocols or
systems~\cite{Bernstein:1986:CCR:17299,Harding:2017:EDC:3055540.3055548}.
Instead we focus on a few systems that highlight differences in the use of
one-sided RDMA, strictness, serializability, opacity, availability,
and timestamp generation.

Most systems with distributed transactions and data partitioning use
RPCs to read remote objects during execution which requires CPU
processing at the remote participants. Traditional 2-phase commit also
requires processing of $PREPARE$ and $COMMIT$ messages at all
participants, including read-only
ones. Calvin~\cite{Thomson:2012:CFD:2213836.2213838} uses
deterministic locking of predeclared read and write sets to avoid 2PC
but read locks must still be taken using messages.
Sinfonia~\cite{sinfonia} can piggyback reads on the 2PC messages in
specialized cases to reduce the number of messages at read-only
participants.  Sundial~\cite{sundial} uses logical leases to
dynamically reorder transactions to minimize conflicts, and caching to
reduce remote accesses. This provides serializability but not
strictness, and lease renewals still require RPCs to read-only participants.

Systems that use one-sided RDMA reads can avoid CPU processing at
read-only participants. NAM-DB~\cite{namdb} uses RDMA reads during
execution and only takes write locks during commit, but it only
provides SI and not
serializability. DrTM~\cite{drtm_sosp,drtm_eurosys} provides
serializability by using hardware transactional memory (HTM) to detect
conflicting writes. FaRM uses an
additional validation phase with RDMA reads to detect conflicting
writes.

With traditional 2PC, if a coordinator fails, the system becomes
unavailable until it recovers. Spanner~\cite{spanner} replicates both
participants and coordinators to provide availability. FaRM and
RAMCloud~\cite{Lee:2015:ILL:2815400.2815416} replicate data but not
coordinators: they recover coordinator state for untruncated
transactions from participants.  In FaRM, transaction recovery is
parallelized across all machines and cores and runs concurrently with
new application transactions. Calvin replicates transactional inputs
and sequence numbers to all nodes in the system and re-executes
transactions issued since the last checkpoint. NAM-DB replicates
inputs of committed transactions but does not checkpoint or replicate
data and must re-execute all past transactions sequentially on failure
before the system becomes available.

Opacity requires consistent read snapshots during execution which can
be provided with pessimistic concurrency control, or with timestamp
ordering.  \farmorig{}, DrTM, RAMcloud, Sundial, and FaSST~\cite{fasst} use
OCC with per-object versions rather than global timestamps, and hence
do not provide opacity. NAM-DB uses timestamp vectors with one element
per server, read from a timestamp server and cached and reused
locally. NAM-DB does not provide strictness. Clock-SI~\cite{clocksi}
uses physical clocks at each server: remote reads must use RPCs that
block at the remote server until the local clock has passed the
transaction read timestamp. Clock-SI does not rely on a clock drift
bound for correctness but needs physical clocks to be loosely
synchronized for performance. It also does not provide strictness.

Spanner~\cite{spanner} and \farmopaque{} use timestamp ordering based
on real time with explicit uncertainty computed according to
Marzullo's algorithm~\cite{marzullo_time}, and both provide opacity. 
Unlike Spanner, \farmopaque{} does not
rely on globally synchronized hardware such as atomic clocks or GPS,
and it achieves two orders of magnitude lower uncertainty than Spanner
within the data center by exploiting fast RDMA networks. Spanner uses
pessimistic concurrency control for serializability whereas \farmopaque{} uses
OCC and supports one-sided RDMA reads.

Timestamp ordering and OCC have been used in many scale-up in-memory
systems, both software transactional memories (STMs) and in-memory
databases. They typically use shared-memory primitives such as CAS to
generate timestamps. TL2~\cite{conf/wdag/DiceSS06} and LSA
~\cite{conf/wdag/RiegelFF06} are some of the first STMs to use
timestamp ordering and OCC to provide strict serializability and
opacity.  TL2 is single-versioned and LSA is multi-versioned with
eager validation. Silo~\cite{silo_sosp13,silo_osdi14} uses OCC with
read-set validation and without opacity for strictly serializable
transactions, and timestamp ordering based on epochs for stale
snapshot reads with opacity. Hekaton~\cite{hekaton_vldb12} uses OCC
and timestamp ordering with both single- and multi-versioning.  It
also supports pessimistic concurrency control.

\section{Conclusion}
\farmopaque{} is a distributed transactional system that provides
opacity, high availability, high throughput, and low latency within a
data center. It uses timestamp ordering based on real time with clocks
synchronized to within tens of microseconds; a protocol to ensure
correctness across clock master failures; and a transactional protocol
that uses one-sided RDMA reads and writes instead of two-sided
messages.  \farmopaque{} achieves 5.4 million neworders/s on a TPC-C
transaction mix and can recover to full throughput within tens of
milliseconds of a failure. To our knowledge this is the highest
throughput reported for a system with opacity and high availability,
which are important to simplify distributed application development.

\section*{Acknowledgements}
We would like to thank Dmitry Bimatov, Paul Brett, Chiranjeeb
Buragohain, Wonhee Cho,  Ming-Chuan Wu, Joshua Cowhig, Dana Cozmei, Orion Hodson, Karthik
Kalyanaraman, Richie Khanna, Alan Lawrence, Ed Nightingale, Greg
O'Shea, John Pao, Knut Magne Risvik, Tim Tan, and Shuheng Zheng for
their contributions to A1 and FaRM. We would also like to thank the anonymous reviewers for
their feedback.

\bibliographystyle{acm}
\bibliography{paper}

\end{document}